\documentclass[a4paper,UKenglish,cleveref, autoref, thm-restate]{lipics-v2021}
\pdfoutput=1 %uncomment to ensure pdflatex processing (mandatatory e.g. to submit to arXiv)
\hideLIPIcs  %uncomment to remove references to LIPIcs series (logo, DOI, ...), e.g. when preparing a pre-final version to be uploaded to arXiv or another public repository

\usepackage{MnSymbol}
\usepackage{todonotes}
\usepackage{tikz}
\usetikzlibrary{automata, positioning, arrows, shapes}

\title{Concurrent Games with Multiple Topologies} %TODO Please add

\author{Shaull Almagor}{Department of Computer Science, Technion, 3200002, Israel}{shaull@cs.technion.ac.il}{0000-0001-9021-1175}{}

\author{Shai Guendelman}
{Department of Computer Science, Technion, 3200002, Israel}
{shaigue@campus.technion.ac.il}
{} % orchid id?
{} % (Optional) author-specific funding acknowledgements
\authorrunning{S. Almagor and S. Guendelman} %TODO mandatory. First: Use abbreviated first/middle names. Second (only in severe cases): Use first author plus 'et al.'

\Copyright{Shaull Almagor and Shai Guendelman} %TODO mandatory, please use full first names. LIPIcs license is "CC-BY";  http://creativecommons.org/licenses/by/3.0/

 %TODO mandatory; please add comma-separated list of keywords

\category{} %optional, e.g. invited paper

\relatedversion{} %optional, e.g. full version hosted on arXiv, HAL, or other respository/website
\nolinenumbers %uncomment to disable line numbering

\EventEditors{Bartek Klin, S\l{}awomir Lasota, and Anca Muscholl}
\EventNoEds{3}
\EventLongTitle{33rd International Conference on Concurrency Theory (CONCUR 2022)}
\EventShortTitle{CONCUR 2022}
\EventAcronym{CONCUR}
\EventYear{2022}
\EventDate{September 12--16, 2022}
\EventLocation{Warsaw, Poland}
\EventLogo{}
\SeriesVolume{243}
\ArticleNo{6}
\newcommand{\cT}{\mathcal{T}}
\newcommand{\cH}{\mathcal{H}}

\newcommand{\eve}{\mathtt{Eve}}
\newcommand{\adam}{\mathtt{Adam}}
\newcommand{\snake}{\mathtt{Snake}}

\newcommand{\parity}{\mathrm{Parity}}

\newcommand{\game}{\mathcal{G}}

\newcommand{\Inf}{\mathrm{Inf}}
\renewcommand{\inf}{\Inf}
\newcommand{\proj}{\mathrm{proj}}
\newcommand{\obey}{\mathrm{obey}}
\newcommand{\nat}{\mathbb{N}}
\newcommand{\tup}[1]{\langle #1 \rangle}
\newcommand{\true}{\mathtt{true}}
\newcommand{\false}{\mathtt{false}}
\newcommand{\substitute}[3]{{#1}[{#2} \mapsto {#3}]}
\newcommand{\players}{\mathsf{Pla}}
\newcommand{\states}{\mathsf{S}}
\newcommand{\actions}{\mathsf{Act}}
\newcommand{\observations}{\mathcal{O}}
\newcommand{\objective}{\alpha}
\newcommand{\transFunc}{\delta}
\newcommand{\topologies}{\mathsf{Top}}
\newcommand{\obs}{\mathrm{obs}}
\newcommand{\strategies}[2]{\Sigma^{#1}_{#2}}
\newcommand{\outcome}[2][]{\mathrm{out}_{#1}({#2})}
\newcommand{\winners}[2][]{\mathrm{Win}_{#1}({#2})}
\newcommand{\wintop}[3][]{\mathrm{WinTop}_{#1}^{#2}({#3})}

\newcommand{\concurrentGame}{\tup{
\players,
\states,
s_0,
\actions,
\transFunc,
(\objective_p)_{p\in \players}
}}

\newcommand{\multiTopologyGame}{\tup{
\players,
\states,
s_0,
\actions,
\topologies,
(\transFunc_t)_{t\in \topologies},
(\objective_{t,p})_{t\in \topologies,p\in \players}
}}

\newcommand{\partialInformationGame}{\tup{
\players,
\states,
s_0,
\actions,
\transFunc,
(\observations_p)_{p\in \players}
}}

\newcommand{\red}{{\color{red} \texttt{red}}}
\newcommand{\blue}{{\color{blue} \texttt{blue}}}
\renewcommand{\vec}[1]{\boldsymbol{#1}}

\begin{document}

\maketitle

\begin{abstract}
Concurrent multi-player games with $\omega$-regular objectives are a standard model for systems that consist of several interacting components, each with its own objective. 
The standard solution concept for such games is Nash Equilibrium, which is a ``stable'' strategy profile for the players. 

In many settings, the system is not fully observable by the interacting components, e.g., due to internal variables. Then, the interaction is modelled by a partial information game. Unfortunately, the problem of  whether a partial information game has an NE is undecidable for the general case.
A particular setting of partial information arises naturally when processes are assigned IDs by the system, but these IDs are not known to the processes. Then, the processes have full information about the state of the system, but are uncertain of the effect of their actions on the transitions. 

We generalize the setting above and introduce Multi-Topology Games (MTGs) -- concurrent games with several possible topologies, where the players do not know which topology is actually used.
We show that extending the concept of NE to these games can take several forms. To this end, we propose two notions of NE: Conservative NE, in which a player deviates if she can strictly add topologies to her winning set, and Greedy NE, where she deviates if she can win in a previously-losing topology.
We study the properties of these NE, and show that the problem of whether a game admits them is decidable.
\end{abstract}

\newpage

\section{Introduction} 
\label{sec:intro}
Concurrent multi-player games of infinite duration over graphs are a standard modelling tool for representing systems that consist of several interacting components, each having its own objective. 
Each player in the game corresponds to a component in the interaction. In each round of the game each of the player chooses an action and the next state of the game is determined by the current state and the vector of actions chosen. A strategy for a player is then a mapping from the history of the game so far to the next action.

A strategy profile (i.e., a tuple of strategies, one for each player) induces an infinite trace of states, and the goal of each player is to direct the game into a trace that satisfies her specification. This is modeled by augmenting the game with $\omega$-regular objectives describing the objectives of the players.

Unlike traditional zero-sum games, here the objectives of the players do not necessarily contradict each other. Accordingly, the typical questions about these games concern their stability. Specifically, the most well-known stability measure is Nash Equilibrium (NE): an NE is a strategy profile such that no single player can improve her outcome by unilaterally deviating from the profile. 
The problem of whether a multi-player game with $\omega$-regular objectives has an NE was shown to be  decidable in~\cite{bouyer2015pure}.

In many settings, the players only have partial information about the system, or can view only certain parts of it. This happens when e.g., the system has private and global variables, and the players model threads that can only view the global variables. 
To this end, games with \emph{partial information} have been extensively studied in various forms~\cite{berthon2021strategy,bouyer2017nash,chatterjee2010complexity,chatterjee2014games}. However, in contrast to the full-information setting, the problem of deciding whether a partial-information multi-player game of infinite duration has a Nash equilibrium is undecidable in the general case where there are 3 or more players~\cite{filiot2018rational} or in the case of stochastic games~\cite{ummels2010complexity}.

In this work, we introduce and study \emph{Multi-Topology Games (MTG)}. Intuitively, an MTG is a concurrent multi-player game with several transition functions (i.e., topologies). Then, players are fully aware of the possible topologies of the game, but do not know which topology they currently play on. Thus, MTGs capture a restricted form of partial information.

As we now demonstrate, MTGs naturally model the sort of partial information that arises in the context of \emph{process symmetry}.
\begin{example}
\label{xmp:process symmetry}
Consider a virtual router with multiple ports. When the router is initialized, several processes are plugged in. The router assigns each process to a port id, but the id is not revealed to the processes. 
Each process attempts to send messages, and its goal is to have its messages delivered (where some messages may be dropped due to heavy traffic).
While the processes know exactly how the router works, they do not know which port they are assigned to. Therefore, their strategies must be oblivious to their port number. 

As a concrete example, consider the concurrent game in \cref{fig:router-game} with players $\{\blue,\red\}$.
When both players know the port assignment, for example, $\blue\to$Port $1$ and $\red\to$Port $2$, then $\blue$ can win by always taking action 1, and $\red$ will lose in any strategy.
However, if the port assignment is not known then in order for either player to win under both port assignments, the players must coordinate e.g., by taking turns trying to send a message.
Thus, a-priori, the game has two possible topologies: \cref{subfig:n1-p1-n2-p2} and \cref{subfig:n1-p2-n2-p1}.

\begin{figure}[ht]
\centering
\begin{subfigure}{0.45\textwidth}
\centering
\begin{tikzpicture} [auto]
    \node (ready) [state, initial below] {$ready$};
    \node (send1) [state, left = of ready, color = blue] {$send_1$};
    \node (send2) [state, right = of ready, color = red] {$send_2$};
    \path [-to]
        (ready) edge [loop above] node {\color{blue}0\color{red}0} ()
        (ready) edge [bend right] node [above] {\color{blue}1\color{red}0\color{black},\color{blue}1\color{red}1} (send1)
        (send1) edge [bend right] node [below] {} (ready)
        (ready) edge [bend right] node [below] {\color{blue}0\color{red}1} (send2)
        (send2) edge [bend right] node [above] {} (ready)
        ;
    \end{tikzpicture}
    \caption{$\blue\to$Port $1$, $\red\to$Port $2$.}
    \label{subfig:n1-p1-n2-p2}
\end{subfigure}
\hfill
\begin{subfigure}{0.45\textwidth}
\centering
\begin{tikzpicture} [auto]
    \node (ready) [state, initial below] {$ready$};
    \node (send1) [state, left = of ready, color = red] {$send_1$};
    \node (send2) [state, right = of ready, color = blue] {$send_2$};
    \path [-to]
        (ready) edge [loop above] node {\color{blue}0\color{red}0} ()
        (ready) edge [bend right] node [above] {\color{blue}0\color{red}1\color{black},\color{blue}1\color{red}1} (send1)
        (send1) edge [bend right] node [below] {} (ready)
        (ready) edge [bend right] node [below] {\color{blue}1\color{red}0} (send2)
        (send2) edge [bend right] node [above] {} (ready)
        ;
    \end{tikzpicture}
    \caption{$\blue\to$Port $2$, $\red\to$Port $1$.}
    \label{subfig:n1-p2-n2-p1}
\end{subfigure}
\caption{Router game from \cref{xmp:process symmetry}.
The players are $\blue$ and $\red$, and the router has two ports $1,2$.
In every round each player can try to \emph{send} (action 1), or \emph{wait} (action 0).
The labels on the edges describe the actions of the players. The first is the action of the $\blue$ player, and the second is the action of the $\red$ player.
From $ready$, if only the player in Port $i\in \{1,2\}$ tries to send, the game transitions to $send_i$. If both players try to send, the router prioritizes the request from Port $1$. 
The objective of the player Port $i$ is to visit $send_i$ infinitely many times. Note that $send_i$ is colored according to the player that tries to reach it in each port assignment.
}
\label{fig:router-game}
\end{figure}
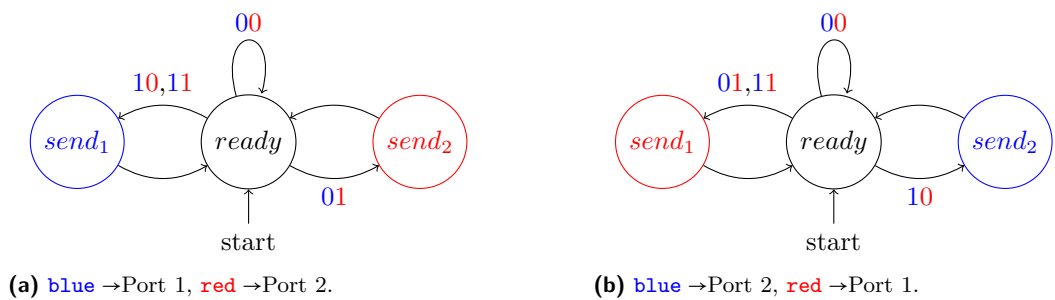

These type of settings are commonly referred to as \emph{process symmetry}~\cite{clarke1996exploiting,emerson1996symmetry,ip1993better,lin2016regular,almagor2020process}, and have been studied in several contexts (e.g., model checking with symmetry reductions). However, to our knowledge this setting has not been studied in games.
In \cref{sec:symmetric_games} we demonstrate how MTGs can model the general setting of process symmetry in games. \hfill \qed
\end{example}

In an MTG, a strategy for a player maps sequences of states to an action, and hence does not depend on a certain topology. Unlike standard games, a strategy profile in an MTG no longer induces a single trace, but rather a set of traces, one per topology. 
Thus, a player can no longer be said to be ``winning'' or ``losing'' in a strategy profile, as this may vary between topologies. In particular, it is not clear how analogues of Nash equilibrium and social optimum should be defined.

To this end, we propose two versions of Nash equilibria, corresponding to two extremities: in a Conservative NE (CNE), a player deviates if she can increase (w.r.t. containment) the set of topologies she wins in. In a Greedy NE (GNE), a player deviates if she can win in a currently-losing topology (even at the cost of losing some of the currently-winning topolgies).

We study the properties of CNE and GNE and compare their strictness, showing that a GNE is also a CNE, but the converse does not hold. We also compare their properties to those of the standard notion of NE.
Our main technical contribution is showing that the problem of whether a game has a CNE (resp. GNE) is decidable.

\subparagraph*{Related Work}
%There are many works relating to concurrent games, partial information games, symmetry and Nash Equilibrium.
A central work concerning NE in concurrent games is~\cite{bouyer2015pure}, where the problem of deciding whether a concurrent game admits an NE was studied for various winning conditions. Apart from establishing tight complexity bounds, this work also introduced the \emph{suspect game} -- a useful technique for reasoning about concurrent games. 
Interestingly, the suspect game does not seem to be adaptable to reason about MTGs, suggesting a fundamental difference between the models.

Zero-sum concurrent reachability games were studied in~\cite{de2007concurrent}, where fundamental techniques for reasoning about them were developed. We remark that the zero-sum setting is technically very different to ours, due to the non-adversarial nature of the players.

Concurrent games can be formulated in the turn-based setting using partial information. The latter were extensively studied, e.g., in~\cite{chatterjee2010complexity, raskin2007algorithms, chatterjee2014games, berthon2021strategy, brenguier2017admissibility, degorre2010energy}, typically in the zero-sum setting.

Finally, the work in~\cite{berthon2021strategy} extends strategy logic~\cite{chatterjee2010strategy} with imperfect information. The authors show that, in general, the model checking problem for this logic is undecidable, but it is decidable in some special cases. Unfortunately, these cases do not readily capture MTGs.

\subparagraph*{Paper organization}
In \cref{sec:preliminaries} we present the basic definitions of concurrent games.
In \cref{sec:MTG} we formally define MTGs, introduce two notions of equilibria for them, and study their properties.
In \cref{sec:solving CNE} we give our main technical result, establishing the decidability of detecting CNE in MTGs. In \cref{sec:solving GNE} we establish the decidability of detecting GNE. 
Finally, in \cref{sec:discussion} we discuss our results and some extensions, and detail future directions.

\section{Preliminaries}
\label{sec:preliminaries}
A \emph{concurrent parity game} is a tuple $\game = \concurrentGame$ where the components are as follows.  
$\players$ is a finite set of players,
$\states$ is a finite set of states,
$s_0\in \states$ is an initial state,
$\actions$ is a finite set of actions. The transition function
$\transFunc: \states \times \actions^\players \to \states$ maps a state and an \emph{action profile} (i.e., $\vec{a} = (a_p)_{p\in \players}\in \actions^\players$) to the next state.
Every player $p\in\players$ has a parity objective $\objective_p \subseteq \states^\omega$, as we describe below.

A \emph{play} of $\game$ is an infinite sequence of states $\rho = s_0, s_1,\ldots \in \states^\omega$ such that for every step $i\in \nat$ there exists an action profile $\vec{a}$ such that $s_{i+1} = \transFunc(s_i,\vec{a})$.
For $k \geq 1$ we denote the length-$k$ prefix of $\rho_{\leq k} = s_0,\ldots,s_{k-1} \in \states^+$.
We denote by $\Inf(\rho)$ the set of states that occur infinitely often in $\rho$.
A \emph{parity objective} is given by a function $\Omega:\states\to \{0,\ldots,d\}$ for some $d\in \nat$. Then, $\rho$ satisfies the objective if $\min\{\Omega(s)\mid s\in \Inf(\rho)\}$ is even. Thus, the objective $\objective_p$ is the set of all plays that satisfy the parity function of Player $p$. In the following, we mostly use the parity function implicitly, and so we do not include $\Omega$ in the description of $\game$. 

The \emph{description size} of $\game$, denoted $|\game|$ is the number of bits required to represent the components of $\game$.
\begin{remark}[Game representation]
\label{rmk:transition_representation}
Note that we assume an explicit representation of the transition function as a table. In particular, we describe for every state the transition on every action profile in $\actions^\players$. Thus, the size of the transition functions is exponential in $|\players|$.

This is in contrast with a more succinct representation, i.e., representing the transition function as a circuit. We choose this focus to eliminate the complexity effect of succinct representation.
\end{remark}

A \emph{history} of $\game$ is a finite prefix of a play $h\in \states^+$.
A \emph{strategy} for Player $p$ is a function $\sigma : \states^+ \to \actions$ that maps a history to the next action of Player $p$. A \emph{strategy profile} $\vec{\sigma} = (\sigma_p)_{p\in \players}$ is vector of strategies, one for each player.
We denote the set of all strategies by $\strategies{}{\game}$ and the set of all strategy profiles by $\strategies{\players}{\game}$ (we omit the subscript $\game$ when it is clear from context).
A strategy profile $\vec{\sigma}$ can be thought as a function that maps histories to action profiles: given a history $h \in \states^+$ we have  $\vec{\sigma}(h) = (\sigma_p(h))_{p\in \players} \in \actions^\players$. 

For a strategy profile $\vec{\sigma}$ we define its \emph{outcome} to be the infinite sequence of states (i.e. play) in $\game$ that is taken when all the players follow their strategies in $\vec{\sigma}$. Formally,  $\outcome[\game]{\vec{\sigma}} = s_0 s_1 \ldots \in \states^\omega$ where $s_0$ is the initial state, and for every $i\ge 1$ we have $s_i=\transFunc(s_{i-1},\vec{\sigma}(s_0,\ldots, s_{i-1}))$. 
Consider a play $\rho\in \states^\omega$. The set of \emph{winners} in $\rho$ is the set of players whose objectives are met in $\rho$. Formally, 
$\winners[\game]{\rho} = \{p \in \players \mid \rho \in \objective_p\} \subseteq \players$. The set of winners in a strategy profile $\vec{\sigma}$ is then $\winners[\game]{\vec{\sigma}} = \winners[\game]{\outcome[\game]{\vec{\sigma}}}$.
Player $p$ is said to be \emph{losing} if she is not winning.

\begin{remark}[Action visibility]
\label{rmk:action_visibility}
Note that strategies are defined to ``see'' only the history of visited states, and not the history of actions taken by the other players. This is a standard and natural assumption~\cite{bouyer2015pure,chatterjee2014games} for concurrent models. 
There are, however, works (e.g.,~\cite{almagor2015repairing}) where players can view the entire action history. The latter approach is slightly easier to reason about, as players have full information on the game progress.
\end{remark}

A strategy profile $\vec{\sigma}$ is a \emph{Nash Equilibrium (NE)}  if, intuitively, no single player can benefit from unilaterally changing her strategy. Since the objectives in our setting are binary, ``benefiting'' amounts to moving from the set of losers to the set of winners. We refer to such a change as a \emph{beneficial deviation}. 
Formally, consider a strategy profile $\vec{\sigma}$, a player $p\in \players$ and a strategy $\sigma_p' \in \strategies{}{\game}$ for Player $p$. We denote by $\substitute{\vec{\sigma}}{p}{\sigma_p'}\in \strategies{\players}{}$ the strategy profile obtained from $\vec{\sigma}$ by replacing $\sigma_p$ with $\sigma'_p$. 
Then, $\vec{\sigma}$ is an NE if for every player $p\in \players$ and every strategy $\sigma_p' \in \strategies{}{\game}$ for Player $p$, if  $p \in \winners[\game]{\substitute{\vec{\sigma}}{p}{\sigma_p'}}$ then $p \in \winners[\game]{\vec{\sigma}}$. Viewed contrapositively: if $p$ loses when $\game$ is played with $\vec{\sigma}$, then $p$ also loses after changing her strategy.

\section{Multi-Topology Games}
\label{sec:MTG}
A \emph{multi-topology game (MTG)} is a tuple $\game=\multiTopologyGame$ where $\players$, $\states$, $s_0$, $\actions$, are the same as in concurrent games. $\topologies$ is a finite set of \emph{topologies}, and for every $t\in \topologies$ we have a transition function $\transFunc_t : \states \times \actions^\players \to \states$ and objective $\objective_{t,p}\subseteq \states^\omega$ for every player $p\in \players$. 
An MTG can be thought of as a tuple of games over the same states, players and actions. That is, for $t\in \topologies$, we can define  $\game_t=\tup{\players,\states,s_0,\actions,\transFunc_t,(\objective_{t,p})_{p\in \players}}$ to be  the concurrent parity game obtained by fixing the transition function to $\transFunc_t$ and the objective for Player $p$ to $\objective_{t,p}$.

Crucially, the players are assumed to have no a-priori information on which topology is selected when the game is played. This is captured in the definition of strategies: a strategy for Player $p$ is identical to the setting of concurrent parity games, i.e., $\sigma_p:\states^+\to \actions$. This lifts to strategy profiles and outcomes, as per \cref{sec:preliminaries}. In particular, a strategy $\sigma$ in $\game$ can be applied to $\game_t$ for every $t\in \topologies$.
Consider a strategy profile $\vec{\sigma}\in \strategies{\players}{}$. The \emph{winning topologies} of Player $p$ is the set of topologies that Player $p$ wins in when $\game$ is played with strategy profile $\vec{\sigma}$. Formally, 
$\wintop[\game]{p}{\vec{\sigma}} = \{
t \in \topologies \mid 
p \in \winners[\game_t]{\vec{\sigma}}\}$.

\subsection{Process Symmetry in Concurrent Games}
\label{sec:symmetric_games}
As we discuss in \cref{sec:intro}, a central motivation for MTGs come from settings where players plug in to the system without knowing their identity. This setting is commonly referred to as \emph{process symmetry}~\cite{clarke1996exploiting,emerson1996symmetry,ip1993better,lin2016regular,almagor2020process}.
Symmetry in games was studied in~\cite{tohme2019structural, stein2011exchangeable, brandt2011equilibria, ham2013notions} for \emph{strategic form games}, which are games with a single turn. 
In~\cite{bouyer2017nash, vester2012symmetric}, symmetry in concurrent games was studied by imposing restrictions on the game structure.
We consider a different setting, where processes $1,\ldots, k$ log into a system described as a concurrent game, but the index of the action controlled by each process is not revealed to the processes. This setting is naturally modelled as an MTG, as follows.

Consider a concurrent game $\game=\concurrentGame$ with $k \geq 2$ players, and that $\players = \{1,\ldots,k\}$. We obtain from $\game$ an MTG with $k!$ topologies by letting each topology correspond to a different permutation of the players. Formally, consider a permutation $\pi \in \mathcal{S}_k$, were $\mathcal{S}_k$ is the set of permutations over $\{1,\ldots,k\}$. For an action profile $\vec{a}\in \actions^\players$ we define $\pi(\vec{a}) = (a_{\pi^{-1}(1)},\ldots,a_{\pi^{-1}(k)})$. That is, the action performed by Player $i$ is taken at index $\pi(i)$. We now obtain the MTG $\game' = \tup{\players,\states,s_0,\actions,\mathcal{S}_k,(\transFunc_\pi)_{\pi\in \mathcal{S}_k},(\alpha_{\pi,p})_{\pi\in\mathcal{S}_k,p\in \players}}$
where $\mathcal{S}_k$ is the set of topologies, 
$\transFunc_\pi$ is obtained by applying $\pi$ to the action profile of the players, that is, for $s\in \states$ and $\vec{a} \in \actions^\players$ we have $\transFunc_\pi(s,\vec{a})=\transFunc(s,\pi(\vec{a}))$.
Finally, the objective of Player $p$ is $\alpha_{\pi,p}=\alpha_{\pi(p)}$.  \cref{fig:router-game} is an example of such game.

\subsection{Solution Concepts}
Recall that in NE, a beneficial deviation moves a player from losing to winning. In MTGs, however, winning is no longer binary. Indeed, a strategy profile associates with each player a set of winning topologies. Thus, the meaning of ``beneficial deviation'' becomes context dependent. We introduce and study two notions of equilibria for MTGs that lie on two ``extremities'': in the \emph{conservative} approach, a deviation is beneficial if it strictly increases (w.r.t. containment) the set of winning topologies. In the \emph{greedy} approach, a deviation is beneficial if a previously-losing topology becomes winning. We now turn to formally define and demonstrate these notions.

\subparagraph*{Conservative NE}
A \emph{conservative NE (CNE)} is a strategy profile $\vec{\sigma}$ where no player can deviate from $\vec{\sigma}$ and have her winning topologies be a strict superset\footnote{we emphasize that the relation $\subsetneq$ means ``strictly contained''.} of her winning topologies when obeying $\vec{\sigma}$.
Formally, $\vec{\sigma}\in \strategies{\players}{}$ is a CNE if the following holds:
\[
\begin{split}
    \forall p \in \players\ 
    \forall \sigma_p' \in \strategies{p}{\game}\ 
    (&(\forall t\in \topologies\  
        p \in \winners[\game_t]{\substitute{\vec{\sigma}}{p}{\sigma_p'}} \to \winners[\game_t]{\vec{\sigma}})
    \lor \\
    & (\exists t\in \topologies\ 
        p \notin \winners[\game_t]{\substitute{\vec{\sigma}}{p}{\sigma_p'}} \land 
        p \in 
        \winners[\game_t]{\vec{\sigma}})
    )
\end{split}
\]
Equivalently, this condition can be written in terms of the set of winning topologies:
\[
%\label{eq:CNE-definition-wintop}
    \forall p\in \players\ \forall \sigma_p'\in \strategies{p}{\game}\  \neg ( 
    \wintop[\game]{p}{\vec{\sigma}} 
    \subsetneq
    \wintop[\game]{p}{\substitute{\vec{\sigma}}{p}{\sigma_p'}})
\]
 
We refer to this notion as \emph{conservative} since a deviating player wants to conserve her existing winning strategies.

\subparagraph*{Greedy NE}
A \emph{greedy NE (GNE)} is a strategy profile $\vec{\sigma}$ where no player can unilaterally deviate and win in a previously-losing topology. Formally, $\vec{\sigma}\in \strategies{\players}{}$ is a GNE if the following holds:
\[
%\label{eq:GNE-definition-explicit}
    \forall p\in \players\ \forall \sigma_p'\in \strategies{p}{\game}\  \forall t\in \topologies\  (p\in \winners[\game_t]{\substitute{\vec{\sigma}}{p}{\sigma_p'}} \to p\in \winners[\game_t]{\vec{\sigma}})
\]
Equivalently, this condition can also be written in terms of the set of winning topologies:
\[
\label{eq:GNE-definition-wintop}
    \forall p\in \players\ 
    \forall \sigma_p'\in \strategies{p}{\game}\ 
    (\wintop[\game]{p}{\substitute{\vec{\sigma}}{p}{\sigma_p'}}
    \subseteq 
    \wintop[\game]{p}{\vec{\sigma}})
\]
The latter formulation shows that in a GNE, for every player and for every deviation, the player's winning topologies when deviating are a subset of the player's winning topologies when obeying $\vec{\sigma}$.
It refer to this notion as \emph{greedy} since it assumes that a player deviates if she improves her outcome in a single topology, disregarding the outcome in other topologies. 
%Thus, a player deviates based on the outcome in single topology.

\begin{example}[CNE and GNE]
\label{example:CNE and GNE}
Recall the router game from \cref{fig:router-game}. The strategy profile where Player $\blue$ repeatedly plays $(0,0,1,1)^\omega$ and $\red$ plays $(1,1,0,0)^\omega$ is a CNE, since the set of winning topologies of this profile is $\{1,2\}$ for both players. Thus, no deviation can win in strictly more topologies.

Note that the same strategy profile is also a GNE, since every set of winning topologies is a subset of $\{1,2\}$.
\end{example}

\begin{remark}[Additional notions of NE]
\label{remark:other-NE-notions}
CNE and GNE are based on the $\subseteq$ preorder on the sets of topologies, $2^\topologies$. In \cref{sec:discussion} we discuss other notions of NE in MTGs.
\end{remark}

\subsection{Properties of  CNE and GNE}
We start by examining some properties and relationships between the notions of CNE and GNE, as well as their relation to standard NE.

Consider an MTG $\multiTopologyGame$. The following observation is immediate from the definitions of GNE and CNE, since if there is only a single topology, the MTG collapses into a concurrent game.
\begin{observation}
\label{observation:single-topology}
If $\topologies=\{t\}$, i.e. there is only a single topology $t$, then the definitions of NE in $\game_t$ coincides with that of CNE and of GNE in $\game$.
\end{observation}

Next, we observe that GNE is a stricter notion than CNE. Indeed, a beneficial deviation in the conservative setting (namely increasing the set of winning topologies) implies a beneficial deviation in the greedy setting (namely winning in a previously-losing topology). Contrapositively, if there is no greedy beneficial deviation, there is also no conservative beneficial deviation. We thus have the following.
\begin{observation}
\label{observation:GNE-implies-CNE}
Let $\game$ be an MTG. If $\vec{\sigma}$ is a GNE in $\game$ then $\vec{\sigma}$ is a CNE in $\game$. 
\end{observation}
The following example shows that the implication of \cref{observation:GNE-implies-CNE} is strict. That is, there are MTGs with a CNE but without a GNE.

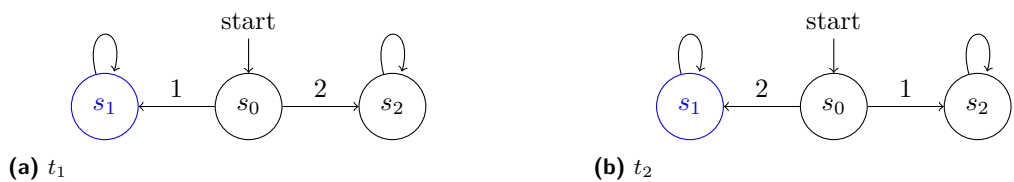
\begin{figure}[ht]
\centering
\begin{subfigure}{0.45\textwidth}
\centering
\begin{tikzpicture} [auto]
    \node (s0) [state, initial above] {$s_0$};
    \node (s1) [state, left = of s0, color = blue] {$s_1$};
    \node (s2) [state, right = of s0] {$s_2$};
    \path [-to]
        (s0) edge node [above] {1} (s1)
        (s1) edge [loop above] node {} ()
        (s0) edge node [above] {2} (s2)
        (s2) edge [loop above] node {} ()
        ;
    \end{tikzpicture}
    \caption{$t_1$}
\end{subfigure}
\hfill
\begin{subfigure}{0.45\textwidth}
\centering
\begin{tikzpicture} [auto]
    \node (s0) [state, initial above] {$s_0$};
    \node (s1) [state, left = of s0, color = blue] {$s_1$};
    \node (s2) [state, right = of s0] {$s_2$};
    \path [-to]
        (s0) edge node [above] {2} (s1)
        (s1) edge [loop above] node {} ()
        (s0) edge node [above] {1} (s2)
        (s2) edge [loop above] node {} ()
        ;
    \end{tikzpicture}
    \caption{$t_2$}
\end{subfigure}
\caption{A single player MTG with two topologies, $t_1$ and $t_2$. In both topologies, the objective of the player is to reach $s_1$ (it is easy to capture this using a parity objective).}
\label{fig:single-agent-CNE-without-GNE}
\end{figure}

\begin{example}[CNE without GNE]
\label{example:CNE without GNE}
Consider the single-player game depicted in \cref{fig:single-agent-CNE-without-GNE}.
The outcome of the game depends only on the first action that the player takes and the topology that the game is played in. 
If the player takes action 1, then the set of winning topologies is $\{t_1\}$.
If the player takes action 2, then the set of winning topologies is $\{t_2\}$.
Since $\{t_1\} \not\subseteq \{t_2\}$ and $\{t_2\} \not\subseteq \{t_1\}$, there is no GNE in the game, as the player can switch strategies from $t_1$ to $t_2$ and vice versa to win in a previously-losing topology.

However, since there is no strategy for the player such that the set of winning topologies is $\{t_1,t_2\}$ (the only strict superset of $\{t_1\}$ and $\{t_2\}$), then every strategy is a CNE.
\end{example}
\begin{remark}[Best-response dynamics in GNE]
\label{rmk:GNE_non_monotonic}
\cref{example:CNE without GNE} demonstrates that, in stark contrast to NE, an MTG might not have a GNE even when there is only a single player. This has to do, in particular, with the notion of \emph{best-response dynamics}: in standard games, one can approach an NE by starting from some profile, and repeatedly letting players deviate to their best-response strategy, until this process converges. While this does not always converge, it does so for a large class of games (e.g., \emph{finite-potential games}~\cite{nisan2007algorithmic}).

Thus, \cref{example:CNE without GNE} shows that best-response does not converge even for a single player in MTGs, whereas it does converge for a single player both for standard NE, as well as in CNE for MTGs. Indeed, the best-response of a single player in the conservative setting will increase her set of winning topologies to the maximum, and from there she will no longer have incentive to deviate.
\end{remark}

\cref{rmk:GNE_non_monotonic} reflects the intuition that a GNE must be stable in each topology separately. That is, it captures the notion ``NE on all topologies'', in the following sense.
\begin{observation}
\label{obs:GNE is NE on all}
A GNE $\vec{\sigma}$ is also an NE in $\game_t$ for every $t\in \topologies$.
\end{observation}
Indeed, if $\vec{\sigma}$ was not an NE in $\game_t$ for some $t\in \topologies$, then a player that deviates from $\vec{\sigma}$ in $\game_t$ would similarly deviate from $\vec{\sigma}$ in $\game$, greedily winning in the previously-losing topology $t$.

In contrast, we now show that CNE is a more intricate notion, and might hold even when there is no NE in the separate topologies. 

\begin{figure}[ht]
\centering
\begin{subfigure}{0.45\textwidth}
\centering
\begin{tikzpicture} [auto]
    \node (s0) [state, initial above] {$s_0$};
    \node (s1) [state, left = of s0, color = blue] {$s_1$};
    \node (s2) [state, right = of s0, color = red] {$s_2$};
    \path [-to]
        (s0) edge node [above] {00,11} (s1)
        (s1) edge [loop above] node {} ()
        (s0) edge node [above] {01,10} (s2)
        (s2) edge [loop above] node {} ()
        ;
    \end{tikzpicture}
    \caption{$t_1$}
\end{subfigure}
\hfill
\begin{subfigure}{0.45\textwidth}
\centering
\begin{tikzpicture} [auto]
    \node (s0) [state, initial above] {$s_0$};
    \node (s1) [state, left = of s0, color = red] {$s_1$};
    \node (s2) [state, right = of s0, color = blue] {$s_2$};
    \path [-to]
        (s0) edge node [above] {00,11} (s1)
        (s1) edge [loop above] node {} ()
        (s0) edge node [above] {01,10} (s2)
        (s2) edge [loop above] node {} ()
        ;
    \end{tikzpicture}
    \caption{$t_2$}
\end{subfigure}
\caption{Symmetric XOR game.
The players are $\blue$ and $\red$.
In topology $t_1$, the objective of $\blue$ is to reach $s_1$, and the objective of $\red$ is to reach $s_2$.
In topology $t_2$ the objectives of the players are swapped.
The game starts from $s_0$.
If both players take the same action, then the game transitions to state $s_1$ and gets stuck there. If the players take different actions then the game transitions to $s_2$ and gets stuck there.}
\label{fig:symmetric-XOR-game}
\end{figure}
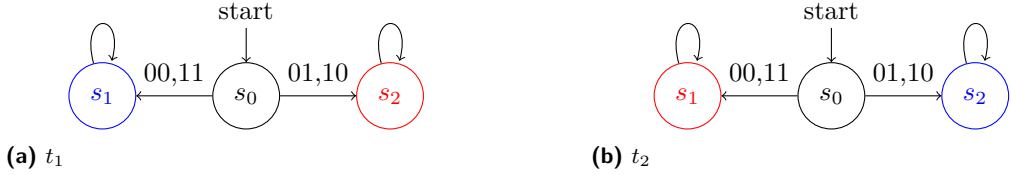

\begin{example}[CNE without NE]
Consider the Symmetric XOR game $\game$ depicted in \cref{fig:symmetric-XOR-game}.
Note that neither $\game_{t_1}$ nor $\game_{t_2}$ have a NE, since if a strategy for a single player is fixed, the other player can respond to it and win. 

On the other hand, any strategy profile is a CNE, since every player always wins in exactly one topology. Thus, there is no way for a player to deviate and get strict superset of winning topologies. 
\end{example}
There are MTGs without CNE. For example, every concurrent game $\game$ without an NE can be viewed as an MTG with a single topology $t_1$. Since there is no NE in $\game$, then for every profile $\vec{\sigma}$ there exists a player $p$ that loses with $\vec{\sigma}$, which corresponds to  $\wintop[\game]{p}{\vec{\sigma}}=\emptyset$ but $p$ can deviate and win $\game$, which corresponds to $\wintop[\game]{p}{\substitute{\vec{\sigma}}{p}{\sigma_p'}}=\{t_1\}$. Since $\emptyset \subsetneq \{t_1\}$, then $\vec{\sigma}$ is not a CNE.

\section{Existence of Conservative NE is Decidable}
\label{sec:solving CNE}
We now turn to our main technical contribution -- showing that the existence of a CNE is a decidable property.

\begin{theorem}
\label{thm:CNE_decidable}
The problem of deciding, given an MTG $\game$, whether there exists a CNE in $\game$ is in 2-EXPTIME.
\end{theorem}
The remainder of the section is devoted to proving \cref{thm:CNE_decidable}. 
Our solution is based on a reduction to the problem of solving a restricted form of partial-information game.
We then employ a result from \cite{chatterjee2014games}, and obtain the complexity result by a careful analysis of the construction.
The rest of the section is organized as follows. In \cref{sec:partial_info} we present the model of partial-information games and the result of \cite{chatterjee2014games}. In \cref{sec:CNE reduction overview} we give an overview of the reduction and in \cref{sec:CNE reduction} we describe and analyze the reduction from our setting.

\subsection{Partial-Information Games}
\label{sec:partial_info}
Partial-information games (also known as \emph{games with incomplete information}) are a ubiquitous model for settings where the players cannot fully observe the state of the game due to e.g., private/hidden variables, unknown parameters or abstractions of part of the system.

Formally, a \emph{partial-information game} is a tuple $\game = \partialInformationGame$ where 
$\players$, $\states$, $s_0$, $\actions$ and 
$\transFunc$ are the same as in concurrent games.
For every player $p\in \players$, the set of \emph{observations} $\observations_p \subseteq 2^\states$ is a partition of $\states$. We omit the acceptance condition, and we will include it explicitly in \cref{thm:krish result} below. 

Intuitively, when the play of $\game$ is at state $s\in \states$, Player $p$ can only observe $o\in \observations_p$ such that $s\in o$, and needs to select an action according to $o$. 
Thus, we distinguish between \emph{state histories}, $\states^+$ and \emph{observation histories (of Player $p$)}, $(\observations_p)^+$.
For $s\in \states$ we define $\obs_p(s)=o\in \observations_p$ to be the unique observation of Player $p$ such that $s\in o$. 
We extend $\obs_p$ to histories: let $h = s_0 s_1 ... s_k \in \states^+$ be a state history, we define $\obs_p(h) = \obs_p(s_0) \obs_p(s_1),\ldots, \obs_p(s_k) \in (\observations_p)^+$ to be the corresponding observation history.

Strategies are \emph{observation based}, that is, a \emph{strategy} for Player $p$ is a function $\sigma_p : \observations_p^+ \to \actions$. Since different players may have different observation sets, we denote by $\strategies{p}{\game}$ the set of all strategies for Player $p$. We denote by $\strategies{\players}{\game}$ the set of all strategy profiles.

Similarly to concurrent games, a strategy profile $\vec{\sigma}$ can be thought of as a function that maps histories to action profiles $\vec{\sigma}(h) = (\sigma_p(\obs_p(h)))_{p\in \players} \in \actions^\players$, and we define $\outcome[\game]{\vec{\sigma}}\in \states^\omega$ similarly to concurrent games.

We say that Player $p\in \players$ has \emph{perfect information} if $\observations_p = \{\{s\} \mid s\in \states\}$. That is, Player $p$ can observe the exact state of the game.
If all players have perfect information then the game is a \emph{perfect information game}, and coincides with our definition of concurrent games. 
We say that \emph{Player $i$ is less informed than Player $j$} if $\observations_j$ is a refinement of $\observations_i$. That is, for every $o_j\in \observations_j$ there exists $o_i \in \observations_i$ such that $o_j \subseteq o_i$.

Finally, consider an objective $\alpha \subseteq \states^\omega$, we say that $\alpha$ is \emph{visible to Player $p$} if for every $\rho,\rho'\in \states^\omega$ such that $\obs_p(\rho)=\obs_p(\rho')$ we have that $\rho\in \alpha$ if and only if $\rho'\in  \alpha$. That is, the objective can be defined according to observation sequences rather than plays.

The following theorem is a result from \cite{chatterjee2014games} that will serve as the target of our reduction.  
\begin{theorem}
\label{thm:krish result}
Let $\game=\partialInformationGame$ be a partial information game, with $\players = \{1,2,3\}$ where Player 1 less informed than Player 2. Let $\objective \subseteq \states^\omega$ be parity objective over $\states$. 
The problem of deciding whether $\exists \sigma_1 \in \strategies{1}{\game}\  \forall \sigma_2 \in \strategies{2}{\game}\  \exists \sigma_3\in \strategies{3}{\game}\  \outcome[\game]{\sigma_1,\sigma_2,\sigma_3}\in \objective$ is 2-EXPTIME complete.
\end{theorem}

\subsection{Overview of the Reduction}
\label{sec:CNE reduction overview}
We now turn to describe a reduction from the CNE existence problem to the setting of \cref{thm:krish result}. 
We start with a high-level description. Consider an MTG $\game$. Instead of asking directly whether $\game$ admits a CNE, we first fix a set of ``intended'' winning topologies $T_p\subseteq \topologies$ for each player $p\in \players$. Then, we ask whether $\game$ admits a CNE $\vec{\sigma}$ in which $\wintop[\game]{p}{\vec{\sigma}}=T_p$ for every $p\in \players$. If we are able to answer the latter problem, we can iterate over every possible tuple $(T_p)_{p\in \players}$ (or nondeterministically guess a set) and conclude whether $\game$ admits a CNE. We remark that this approach is reminiscent of the technique in~\cite{bouyer2015pure}, where the existence of an NE in a game is decided by first guessing a ``witness'' path. 

Once the set of intended topologies is fixed, we construct a 3-player partial-information game whose players are $\eve, \adam$ and $\snake$, with the following roles:
\begin{itemize}
    \item $\eve$ controls the coalition of all players, and suggests a strategy profile $\vec{\sigma}$ by selecting the actions for all the players at each step.
    
    \item $\adam$ selects a deviating player $p$, and the deviating strategy $\sigma'_p$ for that player. In addition, $\adam$ selects a set $T\subseteq \topologies$ in which Player $p$ tries to win when playing $\sigma'_p$. 
    
    \item $\snake$ helps\footnote{It is arguable whether this matches the biblical interpretation. This work makes no theological claims.} $\eve$ by selecting a concrete topology $t$ from the set $T$ picked by $\adam$. 
\end{itemize}
The game starts with $\adam$ and $\snake$ choosing $p$, $T$ and $t\in T$. It then proceeds with $\eve$ and $\adam$ choosing $\vec{\sigma}$ and $\sigma'_p$, respectively, while playing on $\game_t$. The observation sets of the players are such that both $\eve$ and $\adam$ can only observe the current state of the game, so $\eve$ is ignorant of $p$, $T$ and $t$, and $\adam$ is ignorant of $t$ (except knowing that $t\in T$).

The objective of $\eve$ and $\snake$ is then composed of three conditions:
\begin{enumerate}
    \item $\snake$ must choose a topology $t\in T$.
    \item If the strategy $\sigma'_p$ proposed by $\adam$ does not in fact deviate from the profile $\vec{\sigma}$ proposed by $\eve$ (dubbed ``$\adam$ \emph{obeys} $\eve$''), and if $t\in T_p$, i.e., $p$ was intended to win in $t$, then the outcome must be winning for Player $p$.
    \item If $\adam$ selected $T$ to contain a topology not in $T_p$ (i.e., Player $p$ potentially tries to win in a superset of $T_p$), then the outcome must be losing for Player $p$.
\end{enumerate}
The overall idea is that if $\eve$ can find a strategy for all the players, from which any deviation choice of $\adam$ can be shown to be non-beneficial by an appropriate choice by $\snake$, then there is a CNE with the intended winning topologies, and vice-versa. 

There are, however, some caveats: first, in order to allow $\adam$ to choose any set of topologies, the size of the game would be exponential, which is undesirable. Second, it is not immediate that the conjunction of conditions above can be captured by a small parity objective (since the parity condition does not allow conjunction without a change of state space~\cite{boker2018these}). Third, we need to separate the cases where $\adam$ obeys $\eve$. In the following we give the complete construction, which overcomes these caveats.

\subsection{Reduction to Partial Information Game}
\label{sec:CNE reduction}
Consider an MTG $\game=\multiTopologyGame$. For every Player $p\in \players$, fix $T_p \subseteq \topologies$ to be the intended set of winning topologies.

\subparagraph*{Game construction} We construct a 3-player partial-information game $\cH$ with the following components.
The players are $\eve$, $\adam$ and $\snake$. 
The states of $\cH$ are $Q_\cH = \{q_0\}\cup Q$, where $q_0$ is a designated initial state and $Q\subseteq \states \times \players \times 2^\topologies \times \topologies \times \{\true,\false\}$ is described in the following.
A state $(s, p, T, t, b)\in Q$ comprises $s\in \states$ which tracks the state of $\game$, a player $p\in \players$ that is controlled by $\adam$, a set $T\subseteq \topologies$  of topologies that $\adam$ picks, $t\in \topologies$ is a topology picked by $\snake$ and determines the topology $\game$ is played in, and a bit $b\in \{\true,\false\}$ which tracks whether $\adam$ obeys $\eve$.

In order to restrict the state space to a polynomial size in $|\game|$, i.e. reduce the $2^\topologies$ component, we define $\cT_p = \{T_p\cup \{t\}\mid t\in \topologies\} \subseteq 2^\topologies$ and $\cT=(\bigcup_{p\in \players}\cT_p) \cup \{\{t\}\mid t\in \topologies\}$. Note that $|\cT|\le (|\players|+1)\cdot |\topologies|\le  2\cdot |\players|\cdot |\topologies|$.  We now define $Q=\states \times \players\times \cT\times \topologies\times \{\true,\false\}$. Intuitively, the restriction of $2^\topologies$ to $\cT$ is sound, since if a Player $p$ is able to deviate and increase her winning topologies from $T_p$ to some $T$, then she can also increase her winning topologies by just one topology, and thus we can assume $T\in \cT_p$.

We now turn to define the transitions in $\cH$. The actions are defined implicitly by the transitions.\footnote{In the model we describe, actions are identical for all players. However, the model of~\cite{chatterjee2014games} allows different actions as well as enabled and disabled actions in each state, so it is easy to accommodate our actions.}
From $q_0$, $\adam$ selects a player $p\in \players$ and a set of topologies $T \in \cT_p$. As explained in \cref{sec:CNE reduction overview}, $\adam$ controls Player $p$ and attempts to show that $p$ wins in $T$. Still in $q_0$, $\snake$ selects a topology $t\in \topologies$ that $\game$ will be played in. Then, $\cH$ transitions to state $(s_0,p,T,t,\true)\in Q$.

Henceforth, $p,T$ and $t$ remain fixed throughout the play, and $\snake$ has no further effect on the play.
From state $(s, p, T, t, b) \in Q$, $\eve$ chooses an action profile $\vec{a}\in \actions^\players$ %such that for every $i \in \players$, $a_i \in \enabled(s,i)$,
and $\adam$ selects an action $a'_p\in \actions$. Then, the game transitions to state $(s', p, T, t, b') \in Q$ such that $s'=\transFunc_t(s,\substitute{\vec{a}}{p}{a'_p})$,
and $b' = b \land a_p = a'_p$. That is, $\eve$ chooses an action profile, $\adam$ chooses a possible deviation, and the game proceeds according to $\game_t$. If $\adam$ actually deviates, the bit $b$ becomes $\false$ and remains so throughout the play.
Adding $\{\{t\}\mid t\in \cT \}$ to $\cT$ is to make sure that if Player $p$ is supposed to win in topology $t$ (that is, $t\in T_p$), then, the profile suggested by $\eve$ must lead to player $p$ winning in topology $t$. If not, $\adam$ can choose $\{t\}$ and Player $p$ at the start of the game, and obey $\eve$, falsifying one of $\eve$'s winning conditions ($\psi_2$).

Next, we define the observation sets of $\cH$. For a state $q = (s, p, T, t, b) \in Q$ we define the \emph{projection} of $q$ on $\game$ to be $\proj(q) = s$. 
For every state $s\in \states$ of $\game$, let $o_s = \{q\in Q \mid \proj(q) = s\} \subseteq Q$.
The observation sets in $\cH$ are $\observations_\adam=\observations_\eve=\observations=\{\{q_0\}\}\cup \{o_s \mid s\in \states\}$. That is, $\adam$ and $\eve$ can observe the initial state $q_0$, and for every $q\in Q$ they can only observe $\proj(q)$. $\snake$ has perfect information.

This completes the construction of the game $\cH$ (recall that $\cH$ does not have an objective). We proceed to formalize the connection between $\game$ and $\cH$.

\subparagraph*{Correspondence between $\cH$ and $\game$}
We lift the definition of projection to plays: for a play $\rho = q_0 q_1 q_2 ... \in q_0\cdot Q^\omega$ of $\cH$ define $\proj(\rho) = \proj(q_1) \proj(q_2) ...$ (note that we skip the initial state $q_0$). 
We also define the predicate $\obey(\rho) = \bigwedge_{i\ge 1} b_i$, where $b_i$ is the $\true/\false$ bit of $q_i$.
That is, $\obey(\rho)$ is true if and only if $\adam$ always takes the actions suggested by $\eve$. When $\obey(\rho)$ is true, we say that \emph{$\adam$ obeys $\eve$}. 

Since the observation of $\eve$ and $\adam$ correspond to states of $\game$, there is a correspondence between plays, observation-histories and strategies in $\cH$ to plays, histories and strategies in $\game$. We make this precise in the following.
Consider the function $\gamma_\obs : \{q_0\} \cdot \observations^\omega \to \states^\omega$ defined $\gamma_\obs(\{q_0\}, o_{s_0}, o_{s_1}, \ldots) = s_0, s_1, \ldots$. Since $o_{s}=\{q\mid \proj(q)=s\}$ for every $s\in \states$, we have that $\gamma_\obs$ is a bijection between observation-plays of $\eve$ and $\adam$ in $\cH$, and plays of $\game$. By looking at finite sequences, namely histories, we can refer to $\gamma_\obs$ as a bijection between observation-histories of $\adam$ and $\eve$ in $\cH$, and histories in $\game$.
Moreover, since strategies in $\cH$ are observation based, the following functions are also bijective:
\begin{itemize}
    \item $\gamma_\eve : \strategies{\eve}{\cH} \to \strategies{}{\game}$ defined by $\gamma_\eve(\sigma_\eve) = \sigma_\eve \circ \gamma_\obs^{-1}$.
    
    \item $\gamma_\adam : \strategies{\adam}{\cH} \to \bigcup_{p\in \players} \{p\}\times \mathcal{T}_p \times \strategies{p}{\game}$ defined $\gamma_\adam(\sigma_\adam) = (p,T,\sigma_p')$ such that $\sigma_\adam(q_0)=(p,T)$ are the player and the set of topologies selected by $\adam$ in state $q_0$, and $\sigma_p' = \sigma_\adam \circ \gamma_\obs^{-1}$ is the deviating strategy in $\game$ induced by the deviation proposed in $\sigma_\adam$ in $\cH$.
    
    \item $\gamma_\snake : \strategies{\snake}{\cH} \to \topologies$ defined by $\gamma_\snake(\sigma_\snake) = \sigma_\snake(q_0)$ (recall that $\snake$ only acts in $q_0$).
\end{itemize}
For readability, we omit the the subscript and write $\gamma$ instead of $\gamma_\obs,\gamma_\adam,\gamma_\eve,\gamma_\snake$. The correct subscript can be resolved from context. Intuitively, $\gamma$ is the correspondence from strategies/histories/plays in $\cH$ to their counterpart in $\game$.

The connection between strategies and outcomes in $\cH$ and $\game$ is formalized in the following lemma (see \cref{apx:CNE corresponding outcome} for the proof).

\begin{lemma}
\label{lem:CNE corresponding outcome}
    Consider strategies $\sigma_\eve \in \strategies{\eve}{\cH}$, $\sigma_\adam \in \strategies{\adam}{\cH}$ and $\sigma_\snake\in \strategies{\snake}{\cH}$. Let $\vec{\sigma} = \gamma(\sigma_\eve)$, $(p,T,\sigma_p') = \gamma(\sigma_\adam)$ and
     $t = \gamma(\sigma_\snake)$.
    Let $\rho = \outcome[\cH]{\sigma_\eve,\sigma_\adam,\sigma_\snake}$, 
    $\pi' = \outcome[\game_t]{\substitute{\vec{\sigma}}{p}{\sigma_p'}}$,
    and $\pi = \outcome[\game_t]{\vec{\sigma}}$.
    Then 
    $\proj(\rho) = \pi'$.
    Furthermore, if $\adam$ obeys $\eve$ on $\rho$ then
    $\proj(\rho) = \pi = \pi'$.
\end{lemma}

\subparagraph*{Objective for $\cH$}
As sketched in \cref{sec:CNE reduction overview}, the objective $\alpha$ in $\cH$ is constructed so that $\eve$ and $\snake$ can win if and only if there is a CNE in $\game$ with winning topologies $(T_p)_{p\in \players}$. 

We define $\alpha$ as a conjunction of three conditions $\alpha = \{\rho\in q_0\cdot Q^\omega \mid \psi_1(\rho) \land \psi_2(\rho) \land \psi_3(\rho)\}$, where the conditions are defined as follows.
Consider a play $\rho = q_0, (s_0, p, T, t, b_0),  (s_1, p, T, t, b_1),  \ldots$  of $\cH$.
\begin{itemize}
    \item $\psi_1(\rho) := t\in T$. That is, $\psi_1$ forces $\snake$ to choose a topology from the set of topologies selected by $\adam$.
    \item $\psi_2(\rho) := (\obey(\rho) \land t\in T_p) \to \proj(\rho)\in \objective_{t,p}$. That is, $\psi_2$ is satisfied if whenever $\adam$ obeys $\eve$ then Player $p$ wins in any topology $t\in T_p$ selected by $\snake$.
    \item $\psi_3(\rho) := T_p \subsetneq T \to \proj(\rho)\notin \alpha_{t,p}$. That is, $\psi_3$ is satisfied if whenever $\adam$ tries to win in a strict superset of $T_p$, then Player $p$ loses in the topology selected by $\snake$.
\end{itemize}

As mentioned in \cref{sec:CNE reduction overview}, it is not clear that $\alpha$ can be expressed  as a single parity objective over $Q_\cH$.
Nonetheless, we prove that this is possible. 
The key observation is that the ``postconditions'' of $\psi_2$ and $\psi_3$ contradict, hence one of them must hold vacuously. This allows us to decouple the parity conditions for each of them and obtain a single parity objective that captures both, as follows.

For each objective $\alpha_{t,p}$ in $\game$ we write $\alpha_{t,p} = \parity(\Omega_{t,p})$
such that $\Omega_{t,p} : \states \to \{0,\ldots,d\}$ is the parity ranking function, where $d\in \nat$.
We define a new ranking function $\Omega : Q_\cH \to \{0,...,d + 1\}$, and show that $\alpha = \parity(\Omega)$.

First, observe that $q_0$ occurs only once in each play, so its parity rank has no effect. We arbitrarily set $\Omega(q_0)=0$.
Let $\rho \in q_0\cdot Q^\omega$ be a play of $\cH$ and $(s, p, T, t, b),(s', p', T', t', b') \in \Inf(\rho)$. 
It must be that $p=p'$, $T=T'$ and $t=t'$ since those are constant throughout the play, and $b=b'$ since it is either always $\true$ or from some point in $\rho$ it turns into $\false$ and stays that way to the rest of the play.

Let $q = (s,p,T,t,b) \in Q$.
We define $\Omega(q)$ by cases according to $p,T,t,b$, and show that in each case, $\rho \in \alpha$ if and only if $\rho \in \parity(\Omega)$, concluding that $\alpha = \parity(\Omega)$.
For a formula of the form $\psi = \varphi_1 \to \varphi_2$, we refer to $\varphi_1$ as the \emph{precondition} of $\psi$, and $\varphi_2$ as the \emph{postcondition} of $\psi$.
\begin{itemize}
    \item $t \notin T$:
    In this case, if $q\in \inf(\rho)$ then $\rho$ does not satisfy $\psi_1$, thus, $\rho \notin \alpha$. We set $\Omega(q) = 1$ to get $\rho \notin \parity(\Omega)$.
    
    \item $t \in T$, $b=\true$, $t \in T_p$ and $T_p \subsetneq T$:
    In this case, if $q\in \inf(\rho)$ then $\rho$ satisfies the preconditions of both $\psi_2$ and $\psi_3$, but the postconditions of $\psi_2$ and $\psi_3$ contradict, thus, $\rho \notin \alpha$.
    We set $\Omega(q) = 1$ to get $\rho \notin \parity(\Omega)$.

    \item $t \in T$, $b=\true \land t \in T_p$ and $\neg(T_p \subsetneq T)$:
    In this case, if $q \in \inf(\rho)$, then $\rho\in \alpha \iff \proj(\rho)\in \alpha_{t,p}$.
    So we set $\Omega(q) = \Omega_{t,p}(s)$, to apply the objective $\alpha_{t,p}$ over $\proj(\rho)$.
    
    \item $t \in T$, $\neg(b=\true \land t \in T_p)$ and $T_p \subsetneq T$:
    In this case, if $q\in \Inf(\rho)$, then $\rho\in\alpha \iff \proj(\rho)\notin \alpha_{t,p}$.
    So we set $\Omega(q) = \Omega_{t,p}(s) + 1$, to apply the complement of the objective $\alpha_{t,p}$ over $\proj(\rho)$.

    \item $t \in T$, $\neg(b=\true \land t \in T_p)$ and $\neg(T_p \subsetneq T)$:
    In this case, if $q\in \Inf(\rho)$ then $\psi_2$ and $\psi_3$ are vacuously satisfied, and $\rho\in \alpha$. So we set $\Omega(q) = 0$ to get that $\rho\in \parity(\Omega)$.
\end{itemize}

We are now ready to characterize the existence of a CNE in $\game$ by winning strategies in $\cH$.
\begin{lemma}
\label{lem:CNE_reduction_correctness}
Consider an MTG $\game=\multiTopologyGame$. Let $(T_p)_{p\in \players}$ be sets of topologies for each player and let $\cH$ be the corresponding partial-information game.
There exists a strategy profile $\vec{\sigma}$ in $\game$ such that $\vec{\sigma}$ is a CNE and for every $p\in \players$ we have $\wintop[\game]{p}{\vec{\sigma}} = T_p$ 
if and only if the follwing holds:
\[
\exists \sigma_\eve \in \strategies{\eve}{\cH}\ \forall \sigma_\adam \in \strategies{\adam}{\cH}\ \exists \sigma_\snake \in \strategies{\snake}{\cH}\ \outcome[\cH]{\sigma_{\eve}, \sigma_{\adam}, \sigma_{\snake}} \in \objective.
\]
\end{lemma}

\begin{proof}
Assume $\vec{\sigma}$ is a CNE in $\game$ such that for every $p\in \players$, $\wintop[\game]{p}{\sigma}=T_p$, and fix $\sigma_\eve = \gamma^{-1}(\vec{\sigma})$ to be the corresponding strategy for $\eve$ in $\cH$.
Consider a strategy $\sigma_\adam \in \strategies{\adam}{\cH}$ for $\adam$, and let  $(p,T,\sigma_p') = \gamma(\sigma_\adam)$. We show that there exists a strategy $\sigma_\snake\in \strategies{\snake}{\cH}$ so that the outcome satisfies $\alpha$. Recall that a strategy for $\snake$ amounts to choosing a topology. We divide to cases according to the choice of $T$ by $\adam$.
\begin{itemize}
    \item If $\neg (T_p \subsetneq T)$, then $\psi_3$ is satisfied vacuously. 
    Choose $t\in T$ for $\snake$, then $\psi_1$ is satisfied.
    If $\adam$ does not obey $\eve$ or $t \notin T_p$ then $\psi_2$ is vacuously satisfied.
    Otherwise, if $\adam$ obeys $\eve$ and $t\in T_p$, let $\rho = \outcome[\cH]{\sigma_\eve, \sigma_\adam, \sigma_\snake}$. In order to show that $\psi_2$ is satisfied we need to show that $\proj(\rho)\in \alpha_{t,p}$.
    Let $\pi = \outcome[\game_t]{\vec{\sigma}}$. 
    Since $T_p = \wintop[\game]{p}{\vec{\sigma}}$ and $t \in T_p$ we have that $\pi \in \alpha_{t,p}$.
    From \cref{lem:CNE corresponding outcome} we have that $\proj(\rho) = \pi$, so we get that $\proj(\rho)\in \alpha_{t,p}$, as required.

    \item If $T_p \subsetneq T$, denote $T' = \wintop[\game]{p}{\substitute{\vec{\sigma}}{p}{\sigma_p'}}$. Since $\vec{\sigma}$ is a CNE, we have that $\neg(T_p \subsetneq T')$, so $T \setminus T' \neq \emptyset$, as otherwise we would have that $T_p \subsetneq T \subseteq T'$. Choose $t \in T \setminus T'$ for $\snake$, then $\psi_1$ is satisfied.
    Let $\rho = \outcome[\cH]{\sigma_\eve, \sigma_\adam, \sigma_\snake}$,
    $\pi' = \outcome[\game_t]{\substitute{\vec{\sigma}}{p}{\sigma_p'}}$ and $\pi = \outcome[\game_t]{\vec{\sigma}}$.
    From \cref{lem:CNE corresponding outcome} we have that $\proj(\rho) = \pi'$ and if $\adam$ obeys $\eve$ then we have $\proj(\rho) = \pi = \pi'$.
    Note that since $t \notin T' = \wintop[\game]{p}{\substitute{\vec{\sigma}}{p}{\sigma_p'}}$ then $\pi' \notin \alpha_{t,p}$, so $\psi_3$ is satisfied.
    Finally, $\psi_2$ is satisfied vacuously since we cannot have $t \in T_p$ and that $\adam$ obeys $\eve$ simultaneously, as this would yield $T'=T_p=\wintop[\game]{p}{\vec{\sigma}}$, but $t\notin T'$.
\end{itemize}
We conclude that in all cases $\rho\in \alpha$, as required.

Conversely, assume that $\sigma_\eve \in \strategies{\eve}{\cH}$ is such that for every $\sigma_\adam \in \strategies{\adam}{\cH}$ there exists $\sigma_\snake \in \strategies{\snake}{\cH}$ such that $\outcome[\cH]{\sigma_\eve,\sigma_\adam,\sigma_\snake} \in \objective$.
Let $\vec{\sigma} = \gamma(\sigma_\eve)$.
We start by showing that for every $p\in \players$ it holds that $\wintop[\game]{p}{\vec{\sigma}}=T_p$. Indeed, let $p\in \players$ and $t \in \topologies$.

If $t \in T_p$,
take $\sigma_\adam\in \strategies{\adam}{\cH}$ that selects player $p$ and $T=\{t\}$, and obeys $\eve$. 
The only strategy $\sigma_\snake$ for $\snake$ that satisfies $\psi_1$ is to select $t$. Let 
$\rho=\outcome[\cH]{\sigma_{\eve}, \sigma_{\adam}, \sigma_{\snake}}$.
From $\psi_2$ we get that $\proj(\rho) \in \objective_{t,p}$, and by \cref{lem:CNE corresponding outcome} we have $\proj(\rho) = \outcome[\game_t]{\vec{\sigma}}$.
Thus, $t\in \wintop[\game]{p}{\vec{\sigma}}$.

If $t \notin T_p$, take $\sigma_\adam\in \strategies{\adam}{\cH}$ that selects Player $p$ and $T = T_p \cup \{t\}$, and obeys $\eve$.
Since $\adam$ obeys $\eve$, in order for $\psi_1$, $\psi_2$ and $\psi_3$ to be satisfied, $\snake$ must choose $t$, otherwise both preconditions of $\psi_2$ and $\psi_3$ hold, which means that in order to win we must have both $\proj(\rho)\in \alpha_{t,p}$ (by $\psi_2$) and $\proj(\rho)\notin \alpha_{t,p}$ (by $\psi_3$), which cannot hold.
Thus, $\snake$ chooses $t$, and from \cref{lem:CNE corresponding outcome} we have $\proj(\rho)=\outcome[\game_t]{\vec{\sigma}}$. By $\psi_3$ we have $\proj(\rho)\notin \alpha_{t,p}$, so $\outcome[\game_t]{\vec{\sigma}} \notin \objective_{t,p}$. Thus $t\notin \wintop[\game]{p}{\vec{\sigma}}$.
Therefore, $\wintop[\game]{p}{\vec{\sigma}} = T_p$.

It remains to show that $\vec{\sigma}$ is a CNE. 
Assume by way of contradiction that there exists a player $p\in \players$ with a beneficial deviation $\sigma_p'\in \strategies{p}{\game}$. That is, $T' = \wintop[\game]{p}{\substitute{\vec{\sigma}}{p}{\sigma_p'}}$ satisfies $T_p \subsetneq T'$.
We will construct a strategy of $\adam$ such that every strategy of $\snake$ is losing, thereby reaching a contradiction.
Let $T = T_p \cup \{t'\}$ for some $t'\in T\setminus T_p$ and fix $\sigma_\adam = \gamma^{-1}(p,T,\sigma_p')$. 
Consider a strategy $\sigma_\snake$, denote $t = \gamma(\sigma_\snake)$ and let $\rho = \outcome[\cH]{\sigma_\eve,\sigma_\adam,\sigma_\snake}$.
By \cref{lem:CNE corresponding outcome} we have $\proj(\rho)=\outcome[\game_t]{\substitute{\vec{\sigma}}{p}{\sigma_p'}}$, and because $t\in T \subseteq \wintop[\game]{p}{\substitute{\vec{\sigma}}{p}{\sigma_p'}}$ it holds that $\proj(\rho)\in \objective_{t,p}$. However, $T_p\subsetneq T$, so $\psi_3$ is violated, and $\rho \notin \alpha$, which is a contradiction. We conclude that $\vec{\sigma}$ is a CNE. 
\end{proof}

Using \cref{lem:CNE_reduction_correctness} we can decide whether a given MTG $\game$ has a CNE, by iterating over all possible sets of candidate winning topologies $(T_p)_{p\in \players}$, and repeatedly applying the reduction, and using the decision procedure of \cref{thm:krish result}. It remains to analyze the complexity of this procedure.

To this end, observe that the size of $\cH$ is polynomial in the size of $\game$.
Indeed, $|Q|\le |\states|\cdot |\players|\cdot |\cT|\cdot |\topologies|\cdot 2$ where $|\cT|\le 2|\players||\topologies|$. 
and the description of the actions is also polynomial in that of $\game$ (note that $\eve$ has exponentially more actions than each player in $\game$, but the overall description of the transition table in $\game$ is similarly exponential, cf. ~\cref{rmk:transition_representation}).

Finally, by \cref{thm:krish result}, solving $\cH$ takes double-exponential time in $|\game|$, and we have a single-exponential number of iterations, so the overall complexity remains double-exponential time in $|\game|$. This completes the proof of \cref{thm:CNE_decidable}.

\begin{remark}[Lower bounds and improving the upper bound]
\label{rmk:lower bounds}
We do not have a lower bound for the 2-EXPTIME complexity of \cref{thm:CNE_decidable}. Indeed, we suspect that this bound can be lowered. This is due in part to the fact that game $\cH$ we construct does not utilize the full scope of \cref{thm:krish result} from~\cite{chatterjee2014games}. 
Unfortunately, the decision procedure in~\cite{chatterjee2014games} goes through three nontrivial reductions, one of which involves Safra's determinization, that is notoriously difficult to analyze:
The first reduction \cite{chatterjee2010complexity, chatterjee2014games} transforms the objective to a visible objective for $\adam$ which involves the determinization of a parity automaton.
The second reduction \cite{chatterjee2014games} reduces the three-player partial-information game into a two-player partial-information game.
The third reduction uses the results of~\cite{raskin2007algorithms} to reduce the two-player partial-information game to a two-player perfect-information game.

Therefore, it is likely that improving the bound (if indeed possible) will involve devising an ad-hoc procedure, possibly using some key ideas from~\cite{chatterjee2010complexity,chatterjee2014games,raskin2007algorithms}.
\end{remark}

\section{Existence of Greedy NE is Decidable}
\label{sec:solving GNE}
We now turn our attention to Greedy NE (GNE). Recall that a greedy beneficial deviation is one that wins in a previously-losing topology, even at the cost of losing in previously-winning topologies. That is, given an MTG $\game = \multiTopologyGame$, a profile $\vec{\sigma} \in \strategies{\players}{\game}$ is a GNE if for every $p\in \players$, $\sigma_p'\in \strategies{}{\game}$ and $t\in \topologies$, if $p \in \winners[\game_t]{\substitute{\vec{\sigma}}{p}{\sigma'_p}}$ then $p \in \winners[\game_t]{\vec{\sigma}}$.

Intuitively, reasoning in the greedy approach is much less delicate than the conservative approach, since a deviating player need not concern itself with keeping the current winning topologies. As we show in the following, this allows for an exponentially faster solution.

\begin{theorem}
\label{thm:GNE decidable}
The problem of deciding, given an MTG $\game$, whether there exists
a GNE in $\game$ is in EXPTIME.
\end{theorem}
Similarly to \cref{sec:solving CNE}, our approach is to reduce the problem at hand to solving a partial-information game. In the greedy setting, however, it suffices to use two-player games. Specifically, we employ the following result from \cite{chatterjee2010complexity}.
\begin{theorem}
\label{thm:2 player partial information is EXPTIME-complete}
Let $\game = \partialInformationGame$ with $\players = \{1,2\}$. Let $\alpha \subseteq \states^\omega$ be a parity objective.
The problem of deciding whether $\exists \sigma_1\in \strategies{1}{\game}\ \forall \sigma_2\in \strategies{2}{\game}\ \outcome[\game]{\sigma_1,\sigma_2} \in \alpha$ is EXPTIME-complete.
\end{theorem}
We sketch the proof of \cref{thm:GNE decidable}. The complete construction and analysis are detailed in \cref{apx:GNE}.

\begin{proof}[Proof sketch]
As in~\cref{sec:CNE reduction}, we first fix a set of ``intended'' winning topologies $T_p\subseteq \topologies$ for each player $p\in \players$. Then, we ask whether $\game$ admits a GNE $\vec{\sigma}$ in which $\wintop[\game]{p}{\vec{\sigma}}=T_p$ for every $p\in \players$. We then construct a 2-player partial-information game whose players are $\eve, \adam$, where $\eve$ again controls the coalition of all players. 

The behaviour of $\adam$ is different than in the conservative setting. Here, $\adam$ starts by choosing a deviating player $p\in \players$ and a \emph{single topology} $t\in \topologies$ where $p$ attempts to win. The topology $t$ is unobservable by $\eve$. The observations sets of $\eve$ and $\adam$ are again only the current state of $\game$. 
Then, the game is played on topology $t$ with $\eve$ suggesting an action profile, and
$\adam$ possibly deviating with Player $p$. 

The objective for $\eve$ now comprises two conditions: 
\begin{itemize}
    \item $\psi_1$ requires that whenever $\adam$ obeys $\eve$ and $t\in T_p$, the outcome is winning for Player $p$ in $\game_t$.
    \item $\psi_2$ requires that if $t\notin T_p$, then Player $p$ loses in $\game_t$.
\end{itemize}
    
Intuitively, $\adam$ tries to cause Player $p$ to win in a new topology $t$ in which Player $p$ is not intended to win, while $\eve$ is trying to prevent Player $p$ from achieving this, provided that Player $p$ is actually deviating. Note that $\eve$ must do this without knowing which topology is chosen, nor which player deviates (if at all). 
\end{proof}

\section{Discussion, Extensions and Future Work}
\label{sec:discussion}
We introduced MTGs and notions of NE pertaining to them, and showed that deciding whether an MTG admits either notion is decidable (in 2-EXPTIME for CNE and in EXPTIME for GNE). We have also explored the relationships and properties of these notions of NE. We now turn to explore several extensions, and remark about future research directions.

\subparagraph*{Social optimum}
A standard solution concept for concurrent games, apart from NE, is \emph{social optimum}, namely what is the maximum welfare the player can obtain by cooperating. Since in MTGs the winning sets of topologies may be incomparable, we formulate this as follows: given sets $(T_p)_{p\in \players}$, is there a strategy profile $\vec{\sigma}$ such that $\wintop[\game]{p}{\vec{\sigma}}=T_p$ for every $p\in \players$? 

Fortunately, the techniques we developed enable us to readily solve this problem. Indeed, we can modify the reduction used to decide the existence of GNE (\cref{sec:solving GNE}) so that $\adam$ chooses a player and a topology, but does not attempt to deviate and has no further effect on the game. Intuitively, $\adam$ ``challenges'' $\eve$ to show that the winning topologies for the players are exactly the intended ones. The complexity of this approach remains EXPTIME. 

\subparagraph*{Lower bounds} As discussed in~\cref{rmk:lower bounds}, we do not provide lower bounds for our results. Trivial lower bounds on the existence of CNE and GNE can be obtained from those of NE existence in concurrent games, namely $\text{P}^{\text{NP}}_{||}$-hardness~\cite{bouyer2015pure}. This, however, is unlikely to be tight. A central open challenge is to determine the exact complexity of CNE and GNE existence in MTGs. 

\subparagraph*{Additional notions of equilibria} The notions we propose, namely CNE and GNE, lie on two extremities: in the conservative setting a deviation is very strict, and in the greedy setting it is very lax. Generally, one can obtain a notion of equilibrium using any binary relation on $2^\topologies$, which describes what the beneficial deviations are for each player. Moreover, different players can have different relations.

Of particular interest is a quantitative notion of NE, whereby a player deviates if she can increase the \emph{number} of her winning topologies. This notion is fundamentally different from CNE and GNE, as it is not based on set containment, which is key to the correctness of our approach. 

\subparagraph*{Succinct representation of topologies} A central motivation for MTGs, demonstrated in \cref{xmp:process symmetry} and in \cref{sec:symmetric_games} concerns process symmetry. There, from a game with $k$ players, we construct an MTG with $k!$ topologies. However, these topologies can be succinctly represented by computing them on-the-fly. An interesting direction for future work is to determine whether we can devise a symbolic approach that is able to handle such MTGs without incurring an exponential blowup.

\subparagraph*{Logic for partial information games}
Another approach to solve the CNE and GNE existence problems is to formulate those problems with a logic for partial information games~\cite{berthon2021strategy, finkbeiner2010coordination, maubert2014logical}. 
In \cref{sec:SLii} we delve into this approach. As it turns out, while this approach can be described with a more straightforward formula than our solution, the complexity bounds it gives are 3-EXPTIME for both GNE and CNE existence. Moreover, writing the formula essentially requires an understanding of the approach we take in the paper. 
It may be possible to imporove this construction using a more elaborate analysis, but it is not clear what further merit such an analysis will have.

%%
%% Bibliography
%%
%% Please use bibtex, 
\bibliography{multi_topology_games}

\begin{thebibliography}{10}

\bibitem{almagor2020process}
Shaull Almagor.
\newblock Process symmetry in probabilistic transducers.
\newblock In {\em 40th IARCS Annual Conference on Foundations of Software
  Technology and Theoretical Computer Science}, 2020.

\bibitem{almagor2015repairing}
Shaull Almagor, Guy Avni, and Orna Kupferman.
\newblock Repairing multi-player games.
\newblock In {\em 26th International Conference on Concurrency Theory (CONCUR
  2015)}. Schloss Dagstuhl-Leibniz-Zentrum fuer Informatik, 2015.

\bibitem{berthon2021strategy}
Rapha{\"e}l Berthon, Bastien Maubert, Aniello Murano, Sasha Rubin, and Moshe~Y
  Vardi.
\newblock Strategy logic with imperfect information.
\newblock {\em ACM Transactions on Computational Logic (TOCL)}, 22(1):1--51,
  2021.

\bibitem{boker2018these}
Udi Boker.
\newblock Why these automata types?
\newblock In {\em LPAR}, volume~18, pages 143--163, 2018.

\bibitem{bouyer2017nash}
Patricia Bouyer, Nicolas Markey, and Steen Vester.
\newblock Nash equilibria in symmetric graph games with partial observation.
\newblock {\em Information and Computation}, 254:238--258, 2017.

\bibitem{bouyer2015pure}
Patricia~P Bouyer, Romain Brenguier, and Nicolas~N Markey.
\newblock Pure nash equilibria in concurrent games.
\newblock {\em Logical methods in computer science}, 2015.

\bibitem{brandt2011equilibria}
Felix Brandt, Felix Fischer, and Markus Holzer.
\newblock Equilibria of graphical games with symmetries.
\newblock {\em Theoretical Computer Science}, 412(8-10):675--685, 2011.

\bibitem{brenguier2017admissibility}
Romain Brenguier, Arno Pauly, Jean-Fran{\c{c}}ois Raskin, and Ocan Sankur.
\newblock Admissibility in games with imperfect information.
\newblock In {\em CONCUR 2017-28th International Conference on Concurrency
  Theory}, volume~85, pages 2--1. Schloss Dagstuhl--Leibniz-Zentrum fuer
  Informatik, 2017.

\bibitem{chatterjee2010complexity}
Krishnendu Chatterjee and Laurent Doyen.
\newblock The complexity of partial-observation parity games.
\newblock In {\em International Conference on Logic for Programming Artificial
  Intelligence and Reasoning}, pages 1--14. Springer, 2010.

\bibitem{chatterjee2014games}
Krishnendu Chatterjee and Laurent Doyen.
\newblock Games with a weak adversary.
\newblock In {\em International Colloquium on Automata, Languages, and
  Programming}, pages 110--121. Springer, 2014.

\bibitem{chatterjee2010strategy}
Krishnendu Chatterjee, Thomas~A Henzinger, and Nir Piterman.
\newblock Strategy logic.
\newblock {\em Information and Computation}, 208(6):677--693, 2010.

\bibitem{clarke1996exploiting}
Edmund~M. Clarke, Reinhard Enders, Thomas Filkorn, and Somesh Jha.
\newblock Exploiting symmetry in temporal logic model checking.
\newblock {\em Formal methods in system design}, 9(1):77--104, 1996.

\bibitem{de2007concurrent}
Luca De~Alfaro, Thomas~A Henzinger, and Orna Kupferman.
\newblock Concurrent reachability games.
\newblock {\em Theoretical computer science}, 386(3):188--217, 2007.

\bibitem{degorre2010energy}
Aldric Degorre, Laurent Doyen, Raffaella Gentilini, Jean-Fran{\c{c}}ois Raskin,
  and Szymon Toru{\'n}czyk.
\newblock Energy and mean-payoff games with imperfect information.
\newblock In {\em International Workshop on Computer Science Logic}, pages
  260--274. Springer, 2010.

\bibitem{emerson1996symmetry}
E~Allen Emerson and A~Prasad Sistla.
\newblock Symmetry and model checking.
\newblock {\em Formal methods in system design}, 9(1):105--131, 1996.

\bibitem{filiot2018rational}
Emmanuel Filiot, Raffaella Gentilini, and Jean-Fran{\c{c}}ois Raskin.
\newblock Rational synthesis under imperfect information.
\newblock In {\em Proceedings of the 33rd Annual ACM/IEEE Symposium on Logic in
  Computer Science}, pages 422--431, 2018.

\bibitem{finkbeiner2010coordination}
Bernd Finkbeiner and Sven Schewe.
\newblock Coordination logic.
\newblock In {\em International Workshop on Computer Science Logic}, pages
  305--319. Springer, 2010.

\bibitem{ham2013notions}
Nicholas Ham.
\newblock Notions of anonymity, fairness and symmetry for finite strategic-form
  games.
\newblock {\em arXiv preprint arXiv:1311.4766}, 2013.

\bibitem{ip1993better}
C~Norris Ip and David~L Dill.
\newblock Better verification through symmetry.
\newblock In {\em Computer Hardware Description Languages and their
  Applications}, pages 97--111. Elsevier, 1993.

\bibitem{lin2016regular}
Anthony~W Lin, Truong~Khanh Nguyen, Philipp R{\"u}mmer, and Jun Sun.
\newblock Regular symmetry patterns.
\newblock In {\em International Conference on Verification, Model Checking, and
  Abstract Interpretation}, pages 455--475. Springer, 2016.

\bibitem{maubert2014logical}
Bastien Maubert.
\newblock {\em Logical foundations of games with imperfect information: uniform
  strategies}.
\newblock PhD thesis, Universit{\'e} Rennes 1, 2014.

\bibitem{nisan2007algorithmic}
Noam Nisan, Tim Roughgarden, {\'{E}}va Tardos, and Vijay~V. Vazirani, editors.
\newblock {\em Algorithmic Game Theory}.
\newblock Cambridge University Press, 2007.
\newblock \href {https://doi.org/10.1017/CBO9780511800481}
  {\path{doi:10.1017/CBO9780511800481}}.

\bibitem{raskin2007algorithms}
Jean-Fran{\c{c}}ois Raskin, Thomas~A Henzinger, Laurent Doyen, and Krishnendu
  Chatterjee.
\newblock Algorithms for omega-regular games with imperfect information.
\newblock {\em Logical Methods in Computer Science}, 3, 2007.

\bibitem{stein2011exchangeable}
Noah~Daniel Stein.
\newblock {\em Exchangeable equilibria}.
\newblock PhD thesis, Massachusetts Institute of Technology, 2011.

\bibitem{tohme2019structural}
Fernando~A Tohm{\'e} and Ignacio~D Viglizzo.
\newblock Structural relations of symmetry among players in strategic games.
\newblock {\em International Journal of General Systems}, 48(4):443--461, 2019.

\bibitem{ummels2010complexity}
M~Ummels and DK~Wojtczak.
\newblock The complexity of nash equilibria in stochastic multiplayer games.
\newblock {\em Logical Methods in Computer Science}, 2010.

\bibitem{vester2012symmetric}
Steen Vester.
\newblock {\em Symmetric Nash Equilibria}.
\newblock PhD thesis, Master’s thesis, ENS Cachan, 2012.

\end{thebibliography}

\appendix

\section{Proofs}
\subsection{Proof of \cref{lem:CNE corresponding outcome}}
\label{apx:CNE corresponding outcome}
We prove by induction that for every $k\geq 1$, $\proj(\rho_{\leq k+1}) = \pi'_{\leq k}$, and if $\adam$ obeys $\eve$ then $\proj(\rho_{\leq k+1}) = \pi'_{\leq k} = \pi_{\leq k}$.
For $k = 1$, $\rho_{\leq 2} = q_0, (s_0,p,t,T,b_0)$ and $\pi'_{\leq 1} = \pi_{\leq 1} = s_0$ and we have that $\proj(\rho_{\leq k+1}) = \pi'_{\leq k}$.
Assuming that $\proj(\rho_{\leq k+1}) = \pi'_{\leq k}$ for $k\geq 1$, the next state of $\proj(\rho)$ will depend on the transition function $\transFunc_t$ and action profile $\substitute{\vec{\sigma}}{p}{\sigma_p'}(\pi'_{\leq k})$ from the way $\gamma$ and the transitions of $\cH$ are defined, and the next state in $\pi'$ will also depend on the same transition function and action profile. Thus, it holds that $\proj(\rho_{\leq k+2}) = \pi'_{\leq k + 1}$.
Farther more, if $\adam$ obeys $\eve$ then in every step the action that $\adam$ takes is identical to the action that $\eve$ suggests for Player $p$, so we have that $\substitute{\vec{\sigma}}{p}{\sigma_p'}(\pi'_{\leq k}) = \vec{\sigma}(\pi'_{\leq k})$, and $\pi_{\leq k + 1} = \pi'_{\leq k + 1}$, thus, $\proj(\rho_{\leq k+2}) = \pi_{\leq k + 1} = \pi'_{\leq k + 1}$. \hfill \qed

\section{Proof of \cref{thm:GNE decidable}}
\label{apx:GNE}
Consider an MTG $\game = \multiTopologyGame$. For every Player $p\in \players$ fix $T_p \subseteq \topologies$ to be the intended set of winning topologies.

\subparagraph*{Game construction}
We construct a two-player partial-information game $\cH$ with the following components. The players are $\eve$ and $\adam$.
The states of $\cH$ are $Q_\cH = \{q_0\} \cup Q$ such that
$q_0$ is a designated initial state 
and $Q = \states \times \players \times \topologies \times \{\true,\false\}$ is described in the following.
A state $(s,p,t,b)\in Q$ comprises of $s\in \states$ which tracks the state of $\game$, a player $p\in \players$ that is controlled by $\adam$, a topology $t\in \topologies$ that $\adam$ picks,
and a bit $b\in \{\true,\false\}$ which tracks whether $\adam$ obeys $\eve$.

We now turn to define the transitions of $\cH$.
The actions are defined implicitly by the transitions.
From state $q_0$, $\adam$ selects a player $p\in \players$ to control
and a topology $t\in \topologies$ that $\game$ will be played in. Then, $\cH$ transitions to state $(s_0,p,t,\true)\in Q$.
Henceforth, $p$ and $t$ remain fixed throughout the play.
From state $(s, p, t, b) \in Q$, $\eve$ chooses an action profile $\vec{a}\in \actions^\players$, and $\adam$ selects an action $a'_p\in \actions$ and $\cH$ transitions to state $(s',p,t,b') \in Q$ such that $s' = \transFunc_t(s,\substitute{\vec{a}}{p}{a'_p})$, and $b' = b \land (a'_p = a_p)$.

The observation sets for the players, $\proj$ and $\obey$ are defined similarly as \cref{sec:CNE reduction}.
Correspondence between $\cH$ and $\game$, $\gamma_\obs, \gamma_\eve$ is defined in the same way as in \cref{sec:CNE reduction}, and $\gamma_\adam : \strategies{\adam}{\cH} \to \bigcup_{p\in \players}\{p\}\times \topologies \times \strategies{p}{\cH}$ is defined for $\gamma(\sigma_\adam) = (p,t,\sigma_p')$ such that $(p,t)$ are the player and topology selected by $\sigma_\adam$ in state $q_0$ and $\sigma_p' = \sigma_\adam \circ \gamma_\obs^{-1}$.

The connection between strategies and outcomes in $\cH$ and $\game$ is formalized in the following lemma whose proof is similar to that of \cref{lem:CNE corresponding outcome}.

\begin{lemma}
\label{lemma:GNE corresponding outcome}
Consider strategies $\sigma_\eve \in \strategies{\eve}{\cH}$ and $\sigma_\adam \in \strategies{\adam}{\cH}$.
Let $\vec{\sigma} = \gamma(\sigma_\eve)$ and
$(p,t,\sigma_p') = \gamma(\sigma_\adam)$.
Let $\rho = \outcome[\cH]{\sigma_\eve,\sigma_\adam}$
$\pi' = \outcome[\game_t]{\substitute{\vec{\sigma}}{p}{\sigma_p'}}$ and
$\pi = \outcome[\game_t]{\vec{\sigma}}$.
Then,
$\proj(\rho) = \pi'$.
Furthermore, if $\adam$ obeys $\eve$ on $\rho$ then $\proj(\rho) = \pi = \pi'$.
\end{lemma}

\subparagraph*{Objective for $\cH$}
Let $\rho = q_0 \cdot (s_0,p,t,b_0) \cdot (s_1,p,t,b_1) \cdot ... $ be a play in $\cH$.
The objective $\alpha$ is such that
$\rho \in \alpha \iff \psi_1(\rho) \land \psi_2(\rho)$, where 
\begin{itemize}
    \item $\psi_1(\rho) := (\obey(\rho) \land t \in T_p) \to \proj(\rho) \in \alpha_{t,p}$.
    \item $\psi_2(\rho) := t \notin T_p \to \proj(\rho) \notin \alpha_{t,p}$.
\end{itemize}
$\alpha$ can be expressed as a parity objective as follows.
For every $t\in \topologies$, $p \in \players$, 
let $\Omega_{t,p} : \states \to \{0,...,d_{t,p}\}$ be the priority function for the parity objective $\alpha_{t,p}$ in $\game$.
We construct a priority function $\Omega : Q_\cH \to \{0,...,d\}$ such that $d = \max\{d_{t,p} + 1 \mid t\in \topologies, p\in \players\}$. We set $\Omega(q_0) = 0$ and
for state $q = (s,p,t,b) \in Q$ we have
\[\Omega(q)=\begin{cases}
\Omega_{t,p}(s) + 1 & t \notin T_p\\
\Omega_{t,p}(s) & b \land t\in T_p\\
\Omega(q) = 0 & \neg b \land t\in T_p
\end{cases}\]
If $t\notin T_p$, then, according to $\alpha$, $\rho \in \alpha$ if and only if $\proj(\rho) \notin \alpha_{t,p}$. This is achieved by adding 1 to $\Omega_{t,p}$ which gives us the complement of $\alpha_{t,p}$.
The case where $\adam$ obeys $\eve$ and $t\in T_p$ is captured in the second case, where $\rho \in \alpha$ if and only if $\proj(\rho)\in \alpha_{t,p}$. This is achieved by setting $\Omega$ to be the same as $\Omega_{t,p}$.
In the last case, non of the preconditions of $\psi_1$ and $\psi_2$ hold, so $\rho \in \alpha$. This is achieved by setting $\Omega$ to 0, such that every such play will satisfy the objective.

\begin{lemma}
\label{lemma:GNE reduction correctness}
There exists a GNE $\vec{\sigma}\in \strategies{}{\game}$ in $\game$ with $\wintop[\game]{p}{\vec{\sigma}}=T_p$ for every $p\in \players$,
if and only if 
$\exists \sigma_\eve\in \strategies{\eve}{\cH}\ \forall \sigma_\adam\in \strategies{\adam}{\cH}\ \outcome[\cH]{\sigma_\eve,\sigma_\adam} \in \alpha$.
\end{lemma}

\begin{proof}
Let $\vec{\sigma} \in \strategies{}{\game}$ be a GNE with $\wintop[\game]{p}{\vec{\sigma}}=T_p$ for every $p\in \players$.
Let $\sigma_\eve \in \strategies{\eve}{\cH}$ be the corresponding strategy for $\vec{\sigma}$, 
and let $\sigma_\adam \in \strategies{\adam}{\cH}$ be some strategy for $\adam$ that corresponds to $(p,t,\sigma_p')$.
Let $\rho = \outcome[\cH]{\sigma_\eve,\sigma_\adam}$.
If $\obey(\rho) \land t\in T_p$, then from \cref{lemma:GNE corresponding outcome} we have that $\proj(\rho) = \outcome[\game_t]{\vec{\sigma}}$, and since $t\in T_p = \wintop[\game]{p}{\vec{\sigma}}$ then $\outcome[\game_t]{\vec{\sigma}} \in \alpha_{t,p}$. Thus, $\psi_1$ is satisfied by $\rho$.
If $t \notin T_p$ then from \cref{lemma:GNE corresponding outcome} we have that $\proj(\rho) = \outcome[\game_t]{\substitute{\vec{\sigma}}{p}{\sigma_p'}}$ and since Player $p$ is losing in $t$ when $\game$ is played with $\vec{\sigma}$ and $\vec{\sigma}$ is a GNE, then $\outcome[\game_t]{\substitute{\vec{\sigma}}{p}{\sigma_p'}} \notin \alpha_{t,p}$. Thus, $\psi_2$ is satisfied and $\rho\in \alpha$. 

Conversely, let $\sigma_\eve \in \strategies{\eve}{\cH}$ be such that for any $\sigma_\adam\in \strategies{\adam}{\cH}$ we have $\outcome[\cH]{\sigma_\eve,\sigma_\adam} \in \alpha$.
Let $\vec{\sigma}\in \strategies{}{\game}$ correspond to $\sigma_\eve$. We show that $\vec{\sigma}$ is a GNE.
First, we show that for every $p\in \players$, $\wintop[\game]{p}{\vec{\sigma}} = T_p$.
Let $t\in \topologies$ and $p\in \players$.
Take $\sigma_\adam \in \strategies{\adam}{\cH}$ that corresponds to $(p,t,\sigma_p)$ where $\sigma_p$ is the strategy assigned to $p$ in $\vec{\sigma}$.
Let $\rho_t = \outcome[\game_t]{\vec{\sigma}}$ and $\rho = \outcome[\cH]{\sigma_\eve,\sigma_\adam}$.
We have that $\rho\in \alpha$.
Since $\adam$ obeys $\eve$ on $\rho$, from \cref{lemma:GNE corresponding outcome} we have that $\proj(\rho) = \rho_t$. If $t\in T_p$ then from $\psi_1$ we get that $\rho_t = \proj(\rho) \in \alpha_{t,p}$, thus, $t\in \wintop[\game]{p}{\vec{\sigma}}$.
If $t\notin T_p$ then from $\psi_2$ we get that $\rho_t = \proj(\rho) \notin \alpha_{t,p}$, thus, $t\notin \wintop[\game]{p}{\sigma}$. So we get that $\wintop[\game]{p}{\sigma} = T_p$.
Now, we show that $\vec{\sigma}$ is a GNE.
Let $p\in \players$, $\sigma_p' \in \strategies{p}{\game}$ and $t\in \topologies$ such that $t \notin T_p$.
Let $\sigma_\adam \in \strategies{\adam}{\cH}$ correspond to $(p,t,\sigma_p')$, and let $\rho = \outcome[\cH]{\sigma_\eve,\sigma_\adam}$. We have that $\rho \in \alpha$, thus, since $t\notin T_p$ then $\proj(\rho)\notin \alpha_{t,p}$.
From \cref{lemma:GNE corresponding outcome} we have that $\rho'_t = \outcome[\game_t]{\substitute{\vec{\sigma}}{p}{\sigma_p'}} = \proj(\rho) \notin \alpha_{t,p}$, thus, $t\notin \wintop[\game_t]{p}{\substitute{\vec{\sigma}}{p}{\sigma_p'}} = T_p$, so $\vec{\sigma}$ is a GNE.
\end{proof}
The algorithm for solving the GNE existence problem is,
for each $(T_p)_{p\in\players}\in (2^\topologies)^\players$ we construct $\cH$ from $\game$ and $(T_p)_{p\in \players}$, and check if there exists $\sigma_\eve\in\strategies{\eve}{\cH}$ such that for every $\sigma_\adam\in \strategies{\adam}{\cH}$, $\outcome[\cH]{\sigma_\eve,\sigma_\adam} \in \alpha$, if there exists such $\sigma_\eve$, then according to \cref{lemma:GNE reduction correctness} is corresponding strategy profile is a GNE, then we return it.
If we went through all $(T_p)_{p\in\players}\in (2^\topologies)^\players$, then return that there does not exist a GNE in $\game$.

The size of $\cH$ is polynomial in the size of $\game$. We copy each $s\in \states$ for every combination of $p\in \players$, $t\in \topologies$, $b\in \{\true,\false\}$, so we get $|Q_\cH| = 2\cdot |S|\cdot |\players| \cdot |\topologies| + 1$, which is polynomial in the size of $\game$. The number of actions in $\cH$ is also polynomial in the number of enabled actions in $\game$ (similarly to the analysis in~\cref{sec:CNE reduction}).

The algorithm performs at most $2^{|\topologies| \cdot |\players|}$ iterations, which is exponential in $|\game|$.
In each iteration we solve $\cH$ with size that is polynomial in $|\game|$, so according to \cref{thm:2 player partial information is EXPTIME-complete} this takes exponential time in $|\game|$, so the GNE existence problem is in EXPTIME.  

% Macros for the next section
% Notation for strategy logic with incomplete information
\newcommand{\slii}{\mathsf{SL_{ii}}}
% Notation for concurrent game structure with imperfect information
\newcommand{\cgsii}{\mathsf{CGS_{ii}}}
% existential strategy quantifier
% first argument - strategy variable
% second argument - observation symbol
\newcommand{\existsStrategy}[2]{\llangle{#1}\rrangle^{#2}}
% universal strategy quantifier
% first argument - strategy variable
% second argument - observation symbol
\newcommand{\forallStrategy}[2]{\lsem{#1}\rsem^{#2}}
% slii strategy binding operator
% first argument is the player
% second argument is the strategy
\newcommand{\bind}[2]{({#1},{#2})}
% slii unbinding operator 
% first argument is the player
\newcommand{\unbind}[1]{({#1},?)}
% temporal operators
\newcommand{\tempX}{\mathrm{X}}
\newcommand{\tempU}{\mathrm{U}}
\newcommand{\tempF}{\mathrm{F}}
\newcommand{\tempG}{\mathrm{G}}
% branching operators
\newcommand{\existsOut}{\mathrm{E}}
\newcommand{\forallOut}{\mathrm{A}}
% definitions for slii
\newcommand{\slAP}{\mathsf{AP}}
\newcommand{\slAgents}{\mathsf{Ag}}
\newcommand{\slVars}{\mathsf{Var}}
\newcommand{\slObsSymbol}{\mathsf{Obs}}
% definitions for cgsii
\newcommand{\slActions}{\mathsf{Ac}}
\newcommand{\slStates}{\mathsf{V}}
\newcommand{\slTrans}{\mathsf{E}}
\newcommand{\slLabel}{\mathcal{L}}
\newcommand{\slInitState}{v_0}
\newcommand{\slObsInterp}{\mathcal{O}}
\newcommand{\cgsiiExplicit}{\tup{\slActions,\slStates,\slTrans,\slLabel,\slInitState,\slObsInterp}}
% others
\newcommand{\topPlayer}{T}
\newcommand{\nondet}{\mathtt{nd}}
\newcommand{\alt}{\mathtt{alt}}

\section{Strategy Logic with Imperfect Information}
\label{sec:SLii}
In this section we discuss solving the GNE and CNE existence problems using \emph{strategy logic with imperfect information}, $\slii$, introduced in~\cite{berthon2021strategy}.
$\slii$ is an expressive logic that is generally undecidable, but a decidable fragment, called \emph{hierarchical instances}, can capture the GNE and CNE existence problems.
The complexity of $\slii$ model-checking for hierarchical instances depends on a parameter called the \emph{simulation depth}.
$\slii$ model-checking for formulas with simulation depth up to $k$ is $(k+1)$-EXPTIME-complete, and the procedure suggested in~\cite{berthon2021strategy} is $(k+1)$-EXPTIME.
Our formulation of the GNE and CNE existence problems with $\slii$, has a simulation depth of 2 for both problems, resulting in 3-EXPTIME procedure for solving those. 
It might be possible that there is a different formulation with a lower simulation depth, lowering the complexity of this approach.

The section is organized as follows. In \cref{sec:slii overview} we give a short overview of $\slii$. In \cref{sec:mtg to cgsii} we discuss how to convert a multi-topology game to a model called \emph{concurrent game structure with imperfect information} that $\slii$ is interpreted over. Then, in \cref{sec:slii GNE} we formalize the GNE existence problem with $\slii$ and compute it's simulation depth. In \cref{sec:slii CNE} we do the same for the CNE existence problem.

\subsection{Overview of $\slii$}
\label{sec:slii overview}
% State the fixed parameters, AP, Ag, Var, Obs
$\slii$ formulas are defined over a number of fixed parameters -- a set of \emph{atomic propositions} $\slAP$, a set of \emph{players} (or \emph{agents}) $\slAgents$, a set of \emph{strategy variables} $\slVars$ and a set of \emph{observation symbols} $\slObsSymbol$.
% Formally define Concurrent Game Structure with II
$\slii$ formulas are interpreted over \emph{Concurrent Game Structure with Imperfect Information}, abbreviated $\cgsii$.
A $\cgsii$ is a tuple $\game = \cgsiiExplicit$ such that 
$\slActions$ is a set of \emph{actions},
$\slStates$ is a set of \emph{states},
$\slTrans: \slStates \times \slActions^\slAgents \to \slStates$ is a \emph{transition function},
$\slLabel: \slStates \to 2^\slAP$ is a \emph{labelling function},
$\slInitState \in \slStates$ is an \emph{initial state} and 
$\slObsInterp: \slObsSymbol \to 2^{\slStates \times \slStates}$ is an \emph{observation interpretation}, which maps each observation symbol $o\in \slObsSymbol$ to an equivalence relation over the states $\slObsInterp(o)\subseteq \slStates \times \slStates$. 
% Explain the minimal syntax and it's semantics
$\slii$ has the following \emph{syntax}:
\begin{align*}
\varphi &:= 
    p \mid 
    \neg \varphi \mid 
    \varphi \lor \varphi \mid
    \existsStrategy{x}{o}\varphi \mid
    \bind{a}{x}\varphi \mid 
    \unbind{a} \varphi \mid 
    \existsOut \psi;
    \ p\in \slAP, x\in \slVars, a\in \slAgents \\
\psi &:= 
    \varphi \mid 
    \neg \psi \mid
    \psi \lor \psi \mid
    \tempX \psi \mid
    \psi \tempU \psi;
\end{align*}
Formulas of type $\varphi$ are called \emph{state formulas} and formulas of type $\psi$ are called \emph{path formulas}. 
The boolean and temporal operators $\neg, \lor, \tempX, \tempU$ have their usual semantics.
The syntax is extended with the boolean and temporal operators $\land, \to, \tempF, \tempG$ that can be expressed with the operators already in the syntax.
The \emph{existential strategy quantifier} $\existsStrategy{x}{o}\varphi$ means, ``there exists a strategy $x$ over the observations $\slObsInterp(o)$ that satisfies $\varphi$''.
The syntax is extended with a \emph{universal strategy quantifier} defined  $\forallStrategy{x}{o}\varphi := \neg \existsStrategy{x}{o}\neg \varphi$.
The \emph{binding operator} $\bind{a}{x}$ binds strategy $x$ to player $a$ and the \emph{unbinding operator} $\unbind{a}$ unbinds player $a$ from it's current strategy.
The \emph{existential outcome quantifier} $\existsOut \psi$ means ``there exists an outcome of the current strategy assignment that satisfies $\psi$''. 
The syntax is extended with a \emph{universal outcome quantifier} defined $\forallOut \psi := \neg \existsOut \neg \psi$.
For a full description of the semantics of $\slii$ we refer readers to~\cite{berthon2021strategy}. 

An $\slii$ \emph{instance} is a pair $(\game,\Phi)$ where $\game$ is a $\cgsii$ and $\Phi$ is an $\slii$ state formula. 
In general, $\slii$ is undecidable. But, a fragment called \emph{hierarchical instances} is decidable. An \emph{hierarchical instance} is such that as we go down the syntax tree of the formula, observations only get finer.

The complexity of the model-checking problem for an hierarchical $\slii$ instance $(\game,\Phi)$ depends on the \emph{simulation depth} of $(\game,\Phi)$. The simulation depth is computed recursively on the formula's structure. The complexity of the model-checking procedure for an instance with simulation depth $k$ is $(k+1)$-EXPTIME.  
For a description of how to compute the simulation depth we refer readers to~\cite{berthon2021strategy}. 

\subsection{MTG to $\cgsii$}
\label{sec:mtg to cgsii}
% describe how to convert a multi-topology-game to a concurrent game structure with imperfect information
% write down the definitions that are useful for both GNE and CNE
In this section we show how to translate an MTG to a $\cgsii$ and a set of formulas that describe the players winning conditions.

Let $\game = \multiTopologyGame$ be an MTG. We denote the players $\players = \{p_1\ldots p_n\}$.
% fixed parameters
First, we fix the parameters over which the $\slii$ formulas are defined, $\slAP$, $\slAgents$, $\slVars$ and $\slObsSymbol$.
The set of atomic propositions is such that we can encode each state and each topology with a unique label (a subset of $\slAP$). This will enable us to write the LTL formula $\psi_{t,p}$ for every $t\in \topologies$ and $p \in \players$ which means that the topology $t$ is played and $p$'s objective is satisfied. 
The set of agents is $\slAgents = \players \cup \{\topPlayer\}$ where $\topPlayer$ is the \emph{topology player} that selects the topology. 
The set of strategy variables is $\slVars = \{\sigma_p \mid p\in \players\}\cup \{\sigma_p' \mid p\in \players\}$.
Since all players have the same observation sets (i.e., can observe the state, but not the topology), we only need a single observation symbol $o$.
Note that every $\slii$ instance with a single observation symbol is inherently hierarchical.

The $\cgsii$ that we use is $\cH = \cgsiiExplicit$. 
The actions in $\cH$ are the actions in $\game$ together with actions for $\topPlayer$ that enable him to select the topology in the first turn of the game.
The states of $\cH$ are $\slStates = (\states \times \topologies) \cup \{\slInitState\}$, where $\slInitState$ is the initial state where $\topPlayer$ selects the topology.
The transition function corresponds to the transition function of $\game$, and allowing $\topPlayer$ to select the topology from the initial state $\slInitState$.
The observation symbol $o$ is interpreted such that $\slInitState$ is distinguishable from all other states and $((s,t),(s',t')) \in \slObsInterp(o)$ (that is, $(s,t)$ and $(s',t')$ are indistinguishable) if and only if $s = s'$.

\subsection{Expressing GNE Existence Problem with $\slii$}
\label{sec:slii GNE}
The following formula expresses the GNE existence problem in $\slii$:
\begin{equation*}
\existsStrategy{\vec{\sigma}}{o}\bind{\players}{\vec{\sigma}}
\bigwedge_{p\in \players} \left[
    \forallStrategy{\sigma_p'}{o} \left(
        \bigwedge_{t\in \topologies} \left(
            \existsOut \psi_{t,p} \lor 
            \neg \bind{p}{\sigma_p'} \existsOut \psi_{t,p}
        \right)
    \right)
\right]
\end{equation*}
Where $\existsStrategy{\vec{\sigma}}{o} := \existsStrategy{\sigma_{p_1}}{o}\ldots \existsStrategy{\sigma_{p_n}}{o}$ is a shorthand way of writing ``there exists a strategy profile''. Similarly, $\bind{\players}{\vec{\sigma}} := \bind{p_1}{\sigma_{p_1}}\ldots \bind{p_n}{\sigma_{p_n}}$ is binding the strategy profile to the players.
When all players except for the topology player $\topPlayer$ are bound to a strategy, the formula $\existsOut \psi_{p,t}$ means that player $p$ wins in topology $t$ under the given strategy assignment.
After we quantify over strategy profiles, we require that for every player $p$ in $\game$, every strategy $\sigma_p'$ and every topology $t$, either player $p$ wins topology $t$ when players are assigned strategy profile $\vec{\sigma}$ or player $p$ loses topology $t$ when she changes her strategy to $\sigma_p'$.

\subparagraph*{Simulation depth}
Now, we compute the simulation depth of the instance.
The computation involves two parameters -- first is the current simulation depth $k \in \nat$ and the second is a parameter that can be either $\nondet$ or $\alt$.
The computation is performed according to Section 5.2 in~\cite{berthon2021strategy}.
Quantifying an LTL formula with $\existsOut$ gives the simulation depth $(0,\nondet)$. Thus, $sd \left(\existsOut \psi_{t,p}\right) = (0,\nondet)$.
Binding a strategy to a player does not change the simulation depth, so we have  
$sd \left(\bind{p}{\sigma_p'}\existsOut \psi_{t,p}\right) = (0,\nondet)$.
Negating a formula keeps the current simulation depth the same and sets the second parameter to $\alt$. Thus,
$sd \left(\neg \bind{p}{\sigma_p'}\existsOut \psi_{t,p}\right) = (0,\alt)$.
Taking a disjunction between two formulas results in the maximum of each parameter of the subformulas (where $\nondet < \alt$), thus, 
$sd \left(\existsOut \psi_{t,p} \lor \neg \bind{p}{\sigma_p'} \existsOut \psi_{t,p} \right) = (0, \alt)$.
The conjunction over all the topologies translates into a negation, disjunction and another negation. Since each subformula $\varphi$ has $sd \left( \varphi \right) = (0, \alt)$, we have that:
\begin{equation*}
sd \left(\bigwedge_{t\in \topologies} \left(\existsOut \psi_{t,p} \lor\neg \bind{p}{\sigma_p'} \existsOut \psi_{t,p}\right)\right) = (0,\alt)
\end{equation*}
The universal strategy quantifier translates into a negation that does not change the simulation depth, an existential strategy quantifier that increases the first parameter by 1 and sets the second parameter to $\nondet$ and another negation that sets the second parameter to $\alt$. So we have that 
\begin{equation*}
sd \left(
    \forallStrategy{\sigma_p'}{o} \left(
        \bigwedge_{t\in \topologies} \left(
            \existsOut \psi_{t,p} \lor 
            \neg \bind{p}{\sigma_p'} \existsOut \psi_{t,p}
        \right)
    \right)
\right)
= (1,\alt)
\end{equation*}
Binding the strategy profile to the players has no effect and the universal strategy quantifier increases the first parameter by 1 and sets the second to $\nondet$, thus,
\begin{equation*}
sd \left(
    \existsStrategy{\vec{\sigma}}{o}\bind{\players}{\vec{\sigma}}
    \bigwedge_{p\in \players} \left[
        \forallStrategy{\sigma_p'}{o} \left(
            \bigwedge_{t\in \topologies} \left(
                \existsOut \psi_{t,p} \lor 
                \neg \bind{p}{\sigma_p'} \existsOut \psi_{t,p}
            \right)
        \right)
    \right]
\right)
= (2,\nondet)
\end{equation*}
Making model-checking complexity of the instance to be 3-EXPTIME.

\subsection{Expressing CNE Existence Problem with $\slii$}
\label{sec:slii CNE}
The following formula expresses the CNE existence problem in $\slii$:
\begin{equation*}
\existsStrategy{\vec{\sigma}}{o}\bind{\players}{\vec{\sigma}}
\bigwedge_{p\in \players} \left[
    \forallStrategy{\sigma_p'}{o} \left(
        \left( 
            \bigwedge_{t\in \topologies} \left(
                \existsOut \psi_{t,p} \lor 
                \neg \bind{p}{\sigma_p'} \existsOut \psi_{t,p}
            \right)
        \right) \lor 
        \left(
            \bigvee_{t\in \topologies} \left(
                \existsOut \psi_{t,p} \land 
                \neg \bind{p}{\sigma_p'} \existsOut \psi_{t,p}
            \right)
        \right)
    \right)
\right]
\end{equation*}
The formula for CNE is similar to the formula for GNE. We change the subformula $\bigwedge_{t\in \topologies} \left(\existsOut \psi_{t,p} \lor \neg \bind{p}{\sigma_p'} \existsOut \psi_{t,p} \right)$, which means that for every topology $t$, player $p$ does not improve her outcome by switching to strategy $\sigma_p'$, by taking a disjunction with $\bigvee_{t\in \topologies} \left(\existsOut \psi_{t,p} \land\neg \bind{p}{\sigma_p'} \existsOut \psi_{t,p}\right)$, which means that there exists a topology where player $p$ wins, and loses if she changes her strategy to $\sigma_p'$.

\subparagraph*{Simulation depth}
The simulation depth of the two subformulas $\bigwedge_{t\in \topologies} \left(\existsOut \psi_{t,p} \lor \neg \bind{p}{\sigma_p'} \existsOut \psi_{t,p} \right)$ and $\bigvee_{t\in \topologies} \left(\existsOut \psi_{t,p} \land\neg \bind{p}{\sigma_p'} \existsOut \psi_{t,p}\right)$ is the same and is equal to $(0,\alt)$. Thus, the conjunction of the two results in a formula with simulation depth $(0,\alt)$. The next steps in the computation of the simulation depth are identical to the computations for GNE, making the simulation depth be $(2,\nondet)$ and the model-checking complexity to be 3-EXPTIME.

\end{document}